\documentclass[a4paper,superscriptaddress,twocolumn,aps,pra,showkeys,nofootinbib,10pt,floatfix]{revtex4-1} 

\usepackage{graphicx} 
\usepackage{epstopdf}
\usepackage{color} 
\usepackage{amsmath} 		
\usepackage{amssymb}  
\usepackage{amsthm}
\usepackage{capt-of}
\DeclareMathOperator{\lcm}{lcm}%

\usepackage{tikz}
\usetikzlibrary{arrows,%
                petri,%
                topaths}%

\newtheorem{theorem}{Theorem}[section]
\newtheorem{lemma}[theorem]{Lemma}
\newtheorem{proposition}[theorem]{Proposition}
\newtheorem{corollary}[theorem]{Corollary}

\newenvironment{remark}[1][Remark]{\begin{trivlist}
\item[\hskip \labelsep {\bfseries #1}]}{\end{trivlist}}

\begin{document}

\title{Trees, Forests, and Stationary States of Quantum Lindblad Systems}

\author{Patrick Rooney}
\email{darraghrooney@gmail.com}
\affiliation{Department of Physics, University of Windsor, ON, N9B 3P4, Canada}

\author{Anthony M. Bloch}
\email{abloch@umich.edu}
\affiliation{Department of Mathematics, University of
Michigan, Ann Arbor, MI 48109}

\author{C. Rangan}
\email{rangan@uwindsor.ca}
\affiliation{Department of Physics, University of Windsor, ON,
N9B 3P4, Canada}

\begin{abstract}
In this paper, we study the stationary orbits of quantum Lindblad systems. We show that they can be characterized in terms of trees and forests on a directed graph with edge weights that depend on the Lindblad operators and the eigenbasis of the density operator. For a certain class of typical Lindblad systems, this characterization can be used to find the asymptotic end-states. There is a unique end-state for each basin of the graph (the strongly connected components with no outgoing edges). In most cases, every asymptotic end-state must be a linear combination thereof, but we prove necessary and sufficient conditions under which symmetry in the Lindblad and Hamiltonian operators hide other end-states or stable oscillations between end-states.
\end{abstract} 

\keywords{quantum mechanics, open systems, Lindblad master equation}

\maketitle

\section{Introduction}

As the interest in quantum computation \cite{NielsenChuangBook} and quantum control \cite{DAlessandroBook} has grown, the obstacle of protecting quantum phenomena from the environment \cite{BreuerPetruccioneBook} has taken a crucial role. 
The Lindblad model of open systems, which models a Markovian evolution \cite{Lindblad76}\cite{GoriniKossakowskiSudarshan76} has been widely studied. The prospect of engineering open quantum systems \cite{LloydViola01}\cite{Baconetal01}\cite{Barreiroetal2011} raises the question of what quantum states can be reached for a given system. Of particular importance are the asymptotic states of a system \cite{Spohn1976}\cite{Frigerio1978}\cite{SchirmerWang2010}\cite{BaumgartnerNarnhofer2012}.

One way to characterize the space of density operators is in terms of its unitary orbits; for example, the pure quantum states constitute one such orbit. The dynamics can be decomposed into inter- and intra-orbit dynamics \cite{us_nis2_b}\cite{us_FBC}. It turns out the vector $\Lambda$ of eigenvalues of the density operator, which indexes the space of orbits, obeys a linear ODE $\dot{\Lambda} = \Omega \Lambda$, where $\Omega$ is a (possibly time-dependent) matrix. $-\Omega$ is a weighted Laplacian matrix: it has non-negative diagonal entries, non-positive off-diagonal entries, and its column-sums vanish. Laplacian matrices are often used in a graph-theoretical  context. A weighted digraph has a corresponding Laplacian matrix where the negative edge weights are the off-diagonal entries, and the diagonal entries are sums of all outgoing (or ingoing) edges for each vertex. Laplacians for unweighted, undirected graphs are well-known for their role in Kirchhoff's Matrix-tree theorem \cite{HarrisHirstBook}, which states that the number of spanning trees of the graph is equal to any $(j,j)$-minor divided by the vertex number. More general matrix-tree theorems can be found; in particular, we will require a version \cite{ChaikenKleitman1978} that relates any principal minor of a weighted, directed Laplacian to the forests on the corresponding graph. 

In this paper, we consider the set of stationary orbits: $\{\Lambda: \Lambda \in \ker (\Omega )\}$. Like the minors of $\Omega$, the stationary space can be characterized by trees and forests of the digraph corresponding to $-\Omega$. While $\Omega$ is time-dependent in general, we show that for a natural class of Lindblad systems, there is a basis for which $\Omega$ is constant. This class is the set of generalized permutation matrices: \emph{i.e.} Lindblad operators that have at most one non-zero entry in every row or column. This includes, for example, diagonal operators, the Pauli matrices, generalized Pauli matrices, annihilation and creation operators, and jump operators in the form $|j\rangle\langle k|$. When the Lindblad operators are generalized permutation matrices, it is possible to find all cases of time-independent cases of $\Omega$, which allows us to characterize all possible end-states of the system in terms of trees and forests. 

Interestingly, \cite{Shietal2016} has recently introduced a so-called quantum Laplacian matrix in terms of networks in open quantum systems. This is a fundamentally different object from the one in this paper: the vertices in their graphs are qubits, while ours are eigenstates of the density operator. Their Laplacian matrix is symmetric, acts on the (complex) Hilbert space, and is  built out of special Lindblad operators that have real values (so that the edge weights are real). Our Laplacian matrix is not necessarily symmetric, acts on the simplex of unitary orbits, and is built out of arbitrary Lindblad operators. 

Section II is devoted to preliminaries: (A) the description of how the dynamics splits into that of the orbits and that of the projectors, (B) a discussion of generalized permutation matrices, (C) an outline of results on the end-states of general Lindblad systems, and finally (D) a description of the necessary graph-theoretic concepts. In section III, we prove that the stationary orbits of a Lindblad system are a linear combination of unique orbits on the basins, each of which has coefficients described by trees on a weighted digraph. In section IV, we show that the constraints dictating the aforementioned linear combination can be described by forests on the same graph. In section V, we prove necessary and sufficient conditions for the existence of hidden asymptotic end-states and oscillations, in the case of generalized permutation matrix Lindblad operators. Finally, in section VI, we give three examples, chosen to illuminate the results of sections III, IV, and V.

\section{Preliminaries}

\subsection{Orbit decomposition}

Our starting point is the Lindblad-von Neumann equation describing Markovian dynamics of an open quantum system:
\begin{align*}
\frac{d}{dt}\rho(t) &= [-iH, \rho(t)] + \mathcal{L}_D(\rho(t)) \\
\mathcal{L}_D(\rho) &:= \sum_{\alpha=1}^N \left( L_\alpha \rho L_\alpha^\dagger - \frac{1}{2} \left( L_\alpha^\dagger L_\alpha\rho + \rho L_\alpha^\dagger L_\alpha \right). \right),
\end{align*}
The density operator $\rho$ describes the quantum state. It is positive semi-definite with trace one, properties which are preserved by the dynamics. The Hamiltonian $H$ describes the internal system dynamics; we will take it to be time-independent. The Lindblad operators $L_\alpha$ describe the interaction with the environment.

The eigenvalues and eigenvectors of $\rho$ have a physical interpretation: an eigenvalue $\lambda_j$ is the probability the system can be found in the corresponding eigenvector $|\psi_j\rangle$. We can write down evolution equations for these quantities \cite{us_FBC}\cite{us_nis2_b}. If $\Lambda$ is a vector of eigenvalues, we have the following linear ODE:
\begin{equation}
\frac{d}{dt}\Lambda = \Omega^\pi \Lambda.  \label{dLamb}
\end{equation}
The matrix $\Omega^\pi$ depends on the eigenprojectors $\pi_j$ of $\rho$:
\begin{align*}
\Omega^\pi_{jk} &:= \left\{ \begin{array}{cc}
w_{jk}^\pi, & j \ne k \\
- \sum_{l\ne k} w_{lk}^\pi, &  j = k  .
\end{array} \right. \\
w_{jk}^\pi &:= \sum_{\alpha=1}^N Tr\left(\pi_jL_\alpha\pi_k\right) \\
\rho &= \sum_{j=1}^n \lambda_j \pi_j.
\end{align*}
Note that the column-sums of $\Omega^\pi$ are zero, which is due to the requirement that $\sum_j\lambda_j = 1$. 

The eigenprojectors obey the following equation: 
\begin{align}
\frac{d}{dt}\pi_j =  & \sum_{k \ne j}\Bigg(
i\left(\pi_jH\pi_k - \pi_k H\pi_j\right)  \nonumber \\
& + \sum_{l=1}^n\frac{\lambda_l}{\lambda_k-\lambda_j} \left( \pi_j\mathcal{L}_D(\pi_l)\pi_k + \pi_k\mathcal{L}_D(\pi_l )\pi_j \right)  \Bigg).\label{tpderiv}
\end{align}

We should make a few technical remarks. First, we are treating $\Lambda$ as a vector so that we can write down the linear ODE. In order to do this, we need to choose an ordering of the eigenvalues. The naive way to do this is to enforce $\lambda_j \ge \lambda_{j+1}$ (or something of the kind). However, this results in discontinuities in $\dot{\Lambda}$ and $\pi_j$ at eigenvalue crossings. Instead, if we pick an initial ordering at $t=0$, the future trajectory of $\Lambda$ is unique if we demand differentiability of $\Lambda$ at all future times \cite{KatoBook}.

Even after resolving the ambiguity in $\Lambda$, there remains the question of defining $\pi_j$ at eigenvalue crossings. Notice that the projector ODE is undefined at crossings as well. We require the $\pi_j$'s to be continuous, which resolves the ambiguity. From \cite{KatoBook}, the projectors will not necessarily be differentiable at crossings. However, the sum of the crossing projectors will be, and one can write down the corresponding derivative formula, where $\pi_j$ is replaced by the higher-rank projector corresponding to the crossing eigenvalues \cite{us_FBC}. In this paper, we will largely ignore these technicalities.

We describe $\Lambda$ as a vector, but it does not live on a vector space. Rather it occupies the $n$-simplex $\mathcal{T} := \{(\lambda_1,\dots,\lambda_n): 0 \le \lambda_j \le 1, \sum_j \lambda_j = 1 \}$, which of course is imbedded in the vector space $\subset \mathbb{R}^n$. It is a manifold with boundary, but there are no difficulties at the boundary, as the Lindblad-von Neumann equation ensures trajectories do not leave the simplex.

\subsection{GPM-Lindblad operators}

In this paper we derive a formula for the stationary orbit(s) of $\rho$, for which $\frac{d}{dt}\Lambda = 0$. By ``orbit", we mean a point in the simplex $\mathcal{T}$, since the unitary orbits of the possible density operators can be identified with the vectors $\Lambda$ modulo re-ordering. 

By ``stationary", we mean \emph{instantaneously} stationary: \emph{i.e.} $\Lambda$ such that $\dot{\Lambda} = 0$. The formula that we prove holds for all Lindblad systems, but since there is clear interdependence between $\Lambda$ and the projectors $\pi_j$, it will only give \emph{persistently} stationary orbits (\emph{i.e.} constant solutions to the Lindblad-von Neumann equation) if we have stationary eigenprojectors $\pi_j$ as well. It is difficult to determine such projectors in the general case.

However, for a natural class of Lindblad operators that we refer to as GPM-Lindblad operators, it is possible to determine the stationary eigenprojectors. GPM stands for generalized permutation matrix: a GPM is a matrix where there is at most one non-zero element per column and row.  This class includes several types of Lindblad operators that are used to model open systems: the Pauli operators, the extended Pauli operators, the creation/annihilation operators, all jump operators in the form $|j\rangle\langle k|$, and of course all diagonal matrices.

We must fix an eigenbasis of the Hamiltonian $\{ |j\rangle: j = 1,\dots, n\}$, and in this basis the Lindblad operators will have the following structure: $L_\alpha = A_\alpha D_\alpha$, where the $A_\alpha$'s are permutation matrices and the $D_\alpha$'s are diagonal matrices, with possibly complex diagonal elements. 

We will establish the following notation, which we will need later. For every $L_\alpha$, let $\sigma_\alpha\in S_n$ be the permutation corresponding to $A_\alpha$. Every permutation consists of disjoint simple cycles. Let $M_{\alpha}$ be the number of cycles in $\sigma_\alpha$. Let $N_{\alpha,\gamma}\subseteq \mathbb{Z}_n$ be the domain of the $\gamma$th cycle, so that $\{ N_{\alpha,\gamma}: \gamma = 1,\dots, M_\alpha \}$ is a partition of $\mathbb{Z}_n$, and let $\sigma_{\alpha,\gamma}$ be the restriction of $\sigma_\alpha$ to $N_{\alpha,\gamma}$. Let $m_{\alpha,\gamma} = |N_{\alpha,\gamma}|$ be the length of the cycle in question. We will occasionally suppress the $\alpha$ subscript when it is not important to the discussion.

GPM's typically are required to have \emph{exactly} one nonzero entry per row or column. We have relaxed this requirement to allow columns and rows of all zeros, since we want to include operators like $|j\rangle\langle k|$, $j\ne k$. This however introduces an ambiguity in the choice of $AD$-decomposition that we want to eliminate.
\begin{proposition}\label{GPMzeros}
Let $n_z = n- \textrm{rank}(L)$ be the the number of empty columns (and rows) of a GPM $L$. The decomposition $L=AD$ is unique if and only if $n_z \le 1$. Otherwise, there are $n_z!$ ways to choose $A$ (whereas $D$ is always unique). However, if one imposes the rule that, for any $\gamma$, the set $\{j: j \in N_\gamma, D_j =0 \}$ has at most one element, then the choice of $A$ is unique. It is always possible to impose this rule.
\end{proposition}

\begin{proof}
$D$ is always unique: the $j$th diagonal element of $D$ is the nonzero element in the $j$th column of $L$ if said column is not empty. Otherwise, the relevant diagonal entry in $D$ is zero. 

Now let $g(L)$ be the digraph where there is an edge $\overleftarrow{jk}$ if and only if $L_{jk}\ne 0$ and $j\ne k$. Since there is at most one nonzero element per row or column, every node has at most one ingoing edge, and one outgoing edge. Each component of $g(L)$ must be an isolated node, a simple cycle (every edge has one outgoing and one ingoing edge), or a unary tree (one root, every node except the last has one daughter). Now construct the permutation $\sigma(g(L))$ as follows. For every isolated node $k$, include the $1$-cycle $(k)$ in $\sigma(g(L))$. For a simple cycle with edges $\overleftarrow{k_1k_2}$, $\dots$ $\overleftarrow{k_{l-1}k_l}$, $\overleftarrow{k_lk_1}$, assign the $l$-cycle $(k_1k_2\dots k_l)$. For the unary tree with edges, $\overleftarrow{k_1k_2}$, $\dots$ $\overleftarrow{k_{l-1}k_l}$, also assign the $l$-cycle $(k_1k_2\dots k_l)$. This permutation is well-defined, since every node is accounted for in some cycle, and all cycles are disjoint. 

Since the permutation is well-defined, the choice of $A$ corresponding to it is unique. Moreover, there can only one zero element in $D$ for each cycle. The isolated nodes correspond to $1$-cycles, so there is no room for more than one zero. The simple cycles in $g(L)$ correspond to cycles with no zeros, and the unary cycles in $g(L)$ correspond to precisely one zero (the missing edge from root to daughterless node).
\end{proof}
When the Lindblad operators are GPM's, the eigenprojectors $\pi_j^D: |j\rangle\langle j|$, $j=1, \dots, n$ are stationary. $H$ is diagonal in this basis, so $\pi^D_j H\pi^D_k=\pi_kH\pi_j^D=0$ for $j\ne k$. The Lindblad super-operator is also diagonal:
\begin{align*}
\mathcal{L}(\pi_j^D) &= \sum_{\alpha=1}^N |D_{\alpha,j}|^2 \left(\pi_{\sigma_{\alpha}(j)}^D - \pi_j^D \right).
\end{align*}
If follows that any $\rho$ that is diagonal remains diagonal. The question of whether there are other stationary eigenprojectors, we leave until section V: the short answer is that, under certain symmetries in the Lindblad operators, there are; otherwise, the diagonal projectors are the only stationary ones.

\subsection{Stationary states of Lindblad systems}

Work has been done by Baumgartner and Narnhofer \cite{BaumgartnerNarnhofer2008}\cite{BaumgartnerNarnhofer2012} regarding the structure of the asymptotic behavior of general Lindblad systems, which we summarize here. The Hilbert space $\mathcal{H}$ can always be decomposed into a direct sum:
\begin{align*}
\mathcal{H} &= \mathcal{D} \oplus \mathcal{R} \\
\mathcal{R} &= \bigoplus_{j} \mathcal{R}_j,
\end{align*}
where the first decomposition is unique. $\mathcal{D}$ is called the decay space: $P_\mathcal{D}\rho(t)P_\mathcal{D} \rightarrow 0$ for any initial condition (where $P_\mathcal{A}$ represents a projector onto a subspace $\mathcal{A}$). The decay space is also maximal, in that there is no sector of $\mathcal{R}$ where the population in that sector will vanish asymptotically for all initial states.

The spaces $\mathcal{R}_j$ are called dissipation blocks: if $\rho(0) = P_{\mathcal{R}_j}\rho(0)P_{\mathcal{R}_j}$, then $\rho(t)$ approaches a unique operator $\rho_{\infty,j}$ that is full-rank on $\mathcal{R}_j$. The projector $P_{\mathcal{R}_j}$ must commute with $P_\mathcal{R}HP_\mathcal{R}$ as well as all $P_\mathcal{R}L_\alpha P_\mathcal{R}$. Any sum of dissipation blocks is an \emph{enclosure}: if an initial density operator has support within any enclosure, its trajectory at all forward times will also have support within that space.

If there are two or more blocks that are unitarily equivalent, then either the decomposition of these blocks is not unique, or there can be oscillations between the two blocks. By ``unitary equivalence", we mean there exist constants $h_j$ and $h_k$, and a matrix $U$ that obeys $U^\dagger U = P_{\mathcal{R}_j}$, $U U^\dagger = P_{\mathcal{R}_k}$, such that $U+U^\dagger$ commutes with $P_{\mathcal{R}}HP_{\mathcal{R}} - h_jP_{\mathcal{R}_j} - h_kP_{\mathcal{R}_k}$ and $P_{\mathcal{R}}L_\alpha P_{\mathcal{R}}$, $\forall \alpha: \alpha = 1, \dots, N$. If this property exists, then:
\begin{enumerate}
\item If $h_j = h_k$, there is a continuous transformation of dissipation blocks. Any matrix $\tilde{\rho}\otimes\rho_{\infty}$ is a possible asymptotic state, where $\tilde{\rho}$ is an arbitrary density operator on $\mathbb{C}^2$, and $\rho_\infty$ is unitarily equivalent to $\rho_{\infty,j}$ and $\rho_{\infty,k}$.
\item If $h_j \ne h_k$, there are stable oscillations between dissipation blocks. The trajectories $R(-\omega_{jk}t)\tilde{\rho} R(\omega_{jk}t)\otimes\rho_\infty$ obey the Lindblad-von Neumann equation, where $\tilde{\rho}$ is an arbitrary density operator on $\mathbb{C}^2$, $R(\theta) := diag(1,e^{i\theta})$, and $\omega_{jk}:=h_j-h_k$. 
\end{enumerate}

Any initial state of the system must approach a linear combination of dissipation states $\rho_{\infty,j}$ and/or oscillations between unitary equivalent dissipation blocks.

\subsection{Trees and Forests}

We now move to graph-theoretical concepts. In this paper, we will connect stationary values of $\Lambda$ to rooted trees and forests on a weighted directed graph. A \emph{tree} is an undirected connected graph without cycles (\emph{i.e.} closed paths). It necessarily has $|V|$ edges, where $V$ is the vertex set. A \emph{rooted} tree is a directed graph whose underlying undirected graph is a tree, and where all edges are oriented towards one particular vertex $r$, called the root. That is, all vertices save $r$ have precisely one outgoing edge, while the root has only incoming edges.  We define a \emph{forest} to be an undirected graph whose connected components are all trees on their respective vertex sets, while a \emph{forest with roots} is a directed graph whose (weakly) connected components are rooted trees. A useful fact regarding trees is Cayley's formula \cite{Cayley}, which states that the possible number of trees on $n$ vertices is $n^{n-2}$. The number of rooted trees with a specified root is then $n^{n-2}$, and the total number of rooted trees with arbitrary root is $n^{n-1}$. Let $T_r(G)$ be the set of sub-trees of a graph $G$ with root $r$ and let $T_{V'}(G)$ be the set of forests whose set of roots is $V' \subseteq V$. These sets may be empty if $G$ is not sufficiently connected.

For a Lindblad system on an $n$-dimensional Hilbert space, we define $G^\pi_\Omega$ to be the graph with $n$ vertices $V(G^\pi_\Omega)=\{v_j: j = 1,2,\dots,n\}$, each vertex corresponding to one eigenprojector $\pi_j$. For every non-zero $w_{jk}^\pi$, construct an edge $e = \overleftarrow{jk}$ with weight $w(e):= w_{jk}^\pi$. We will suppress the superscript $\pi$ henceforth, leaving the dependence on $\pi$ implicit, in order to reduce clutter.

We also assign weights to trees and forests. For any tree or forest $\tau$, define:
\begin{equation*}
W(\tau) := \prod_{e \in \tau} w(e).
\end{equation*}
Also define $W(\tau) = 1$ if $\tau$ has no edges (\emph{i.e.} each component has only one vertex). It turns out that stationary orbits depend only on the weights of sub-trees. If $n=3$ for example, products such as $w_{12}w_{13}$ and $w_{12}w_{23}$ contribute, but products such as $w_{12}w_{21}$ and $w_{12}w_{32}$ do not, since their indices do not correspond to a tree. 

We define a \emph{basin} of $G_\Omega$ to be any strongly connected component with no outgoing edges to its complement. We will denote the basins $G^{(\eta)}_\Omega$, $\eta=1, \dots, n_B$. Let $N_\eta:= \mathcal{V}(G_\Omega^{(\eta)})$ and 
$N_B:= \cup_{\eta=1}^{n_B} N_\eta$. Let $T_B(G_\Omega)$ be the set of all forests where each member tree is rooted in $N_B(G_\Omega)$. It should be clear that $n_B \ge 1$ regardless of connectivity. In a GPM-system, it turns out that the basins usually (but not always) correspond to the dissipation blocks $\mathcal{R}_j$, and the complement of $N_B$ always corresponds to the decay space $\mathcal{D}$. 

The connection between graphs and linear algebra can be seen with the \emph{matrix-tree theorem} for weighted digraphs \cite{ChaikenKleitman1978}, which we will need in our first proof. Chaiken \cite{Chaiken1982} has given a general theorem for all possible minors of a Laplacian matrix, but we only require a simpler version that only applies to principal minors:
\begin{theorem}
Let $I\subset V(G_\Omega)$, and let $\Omega(-I)$ be the matrix obtained by removing the rows and columns of $\Omega$ corresponding to $I$. Note that $\Omega(-I)$ is not necessarily Laplacian. Then $\det(\Omega(-I)) = (-1)^{|G|-|I|}\sum_{\tau\in T_I(G_\Omega)} W(\tau)$.
\end{theorem}

\section{Stationary Orbits}

Our task in this section is to describe the kernel of $\Omega$, and the first step is to identify the rank of $\Omega$:

\begin{proposition}
The matrix $\Omega$ has rank $n-n_B$, and the kernel has dimension $n_B$.
\end{proposition}

\begin{proof}
First consider the structure of $\Omega$ with respect to the strongly connected components of $G_\Omega$. It is possible to order the rows and columns of the matrix so that the strongly connected components form blocks on the diagonal, and the elements below the block-diagonal are zero. To see why this is true, consider that any digraph $G$ induces a digraph $G_{sc}(G)$ where its nodes are the  strongly connected components of $G$, and there is an edge between two nodes iff there is an edge in $G$ between the two components. This digraph \emph{must} be a tree, because if there was any cycle, that cycle would be strongly connected, which contradicts the assumption that the nodes were the strongly connected components. It follows that $G_{sc}(G_\Omega)$ can be used to order the 
rows and columns of $\Omega$: form blocks out of the strongly connected components, and order the blocks so that any block precedes any blocks upstream from it in $G_{sc}(G_\Omega)$. In this way, all of the basins come first, and any edges between blocks must be above the block-diagonal.

We can conclude the following due to this block-diagonal structure: $rank(G_\Omega) = \sum_{\eta=1}^{n_B}rank(\Omega^{(\eta)}) + \sum_{j=1}^{n_{NB}}rank(\tilde{\Omega}^{(j)})$, where the $\Omega^{(\eta)}$ are the blocks corresponding to the basins, and $\tilde{\Omega}^{(j)}$ are the strongly connected components that are not basins. $n_{NB}$ is the number of non-basin components.

Now use the matrix-tree theorem to deduce the rank of the basin blocks. They cannot be full-rank, because the vector of ones is in the kernel of $\Omega^{(\eta)T}$ (note that the basin blocks are Laplacian matrices, but the non-basin blocks are \emph{not}). If the basin has only one vertex, then the matrix is zero, and has rank zero. If the basin has more than one zero, remove an arbitrary row and column $j$, and apply the matrix-tree theorem. Because of the strong connectedness, there is at least one tree with root $j$ in $G_\Omega^{(\eta)}$, and therefore $-\sum_{\tau\in T_j(G_\Omega^{(\eta)})}W(\tau) < 0$. Since the principal minor is non-zero, and only one row/column was removed, it follows that the rank of $\Omega^{(\eta)}$ is $|G_\Omega^{(\eta)}| - 1$.

Now we show that the non-basin blocks are full-rank. Every non-basin block can be written $\tilde{\Omega}^{(j)} =  \tilde{\Omega}^{(j)}_B + \tilde{\Omega}^{(j)}_D$, where the first matrix is a Laplacian matrix containing only edge-weights corresponding to edges within the strongly connected component, and the second matrix is a diagonal matrix containing only the edge weights (times $-1$) of edges pointing outside the component. Let $I_j$ be the domain of $\tilde{\Omega}^{(j)}$, and let $I_{D,k}:=\{l: w_{lk}\ne 0, k\in I_j, l\ne I_j\}$ be the set of destinations of the outgoing edges from the non-basin. Then we have:
\begin{align*}
|\det(\tilde{\Omega}^{(j)})| &= \sum_{I\in 2^{I_j}, I\ne\emptyset}\left(\sum_{\tau\in T_I(G_{\tilde{\Omega}^{(j)}})} W(\tau)\right)\times \\ & \hspace{1in}\left(\prod_{k\in I} \sum_{i\in I_{D,k} } w_{ik}\right).
\end{align*}
Here we have used the determinant of sums formula $\det(A+B) = \sum_{I\in 2^{\mathbb{Z}^n}} C_I$, where $C$ is the matrix formed by taking columns of $A$ for column indices in $I$, and columns of $B$ for the complement. The columns from $\tilde{\Omega}_D^{(j)}$ only contribute diagonal elements, so we can factor each term. The first factor comes from applying the matrix-tree theorem to the columns from $\tilde{\Omega}_B^{(j)}$. The second factor comes from $\tilde{\Omega}_D^{(j)}$ and is just the product of its relevant diagonal elements. Note that all terms have the same sign: $-$ if $\tilde{\Omega}^{(j)}$ has odd dimension, and $+$ if even. Hence the absolute value. We have also excluded $I=\emptyset$ because the determinant of $\tilde{\Omega}^{(j)}_B$ is zero (it is a Laplacian matrix).

The above determinant must be non-zero. All terms in the sum are non-negative, so all we need to find is a positive term. We know there is at least one out-going edge $w_{ik}$, $i\in I_{D,k}$. Choose $I=\{k\}$. The second factor is positive since it is greater than or equal to $w_{ik}$. The first factor also must be positive, since $T_{k}(G_{\Omega^{(j)}})$ is non-empty (due to the strong-connectedness). Hence we have a positive term, and therefore the non-basin block must have full rank.

Summing up the ranks of all blocks, we get $n-n_B$ for the entire matrix. The kernel has dimension $n$ minus the rank, so it has dimension $n_B$.
\end{proof}

Now that we have the kernel dimension, we can describe the   stationary orbits:

\begin{theorem}
The intersection of $\ker(\Omega)$ with $\mathcal{T}$ is the simplex:
\begin{align*}
\{\sum_{\eta=1}^{n_B} s_\eta\Lambda_{\eta,\infty}&: s_\eta \ge 0, \sum_{\eta=1}^{n_B}s_\eta = 1 \} \\
\Lambda_{\eta,\infty}&:= \frac{\sum_{j\in N_\eta}\sum_{\tau \in T_j(G_\Omega^{(\eta)})}W(\tau) e_j}{\sum_{j\in N_\eta}\sum_{\tau \in T_j(G_\Omega^{(\eta)})}W(\tau)},
\end{align*}
where $e_j$ is the vector with one in the $j$'th place, and zeros elsewhere. 
\end{theorem}

\begin{proof}
In the previous proposition, we showed that each basin has a one-dimensional kernel, and every non-basin has an empty kernel. The total kernel is just the span of $\Lambda_{\eta,\infty}$, provided $\Omega^{(\eta)}\Lambda_{\eta,\infty}=0$. The stationary orbits are just the intersection of the kernel with $\mathcal{T}$, which is all linear combinations of the $\Lambda_{\eta,\infty}$'s such that the coefficients sum to one.

So all we have to do is show that $\Omega^{(\eta)}\Lambda_{\eta,\infty}=0$. An arbitrary row $j$ of $\Omega^{(\eta)}$ times $\Lambda_{\eta,0}$ becomes:
\begin{align}
-\left(\sum_{\tau\in T_j(G_\Omega^{(\eta)})}W(\tau)\right) \left(\sum_{k\in N_\eta ,k\ne j} w_{kj}\right) &+ \nonumber \\
  \sum_{k\in N_\eta ,k\ne j} \left(w_{jk} \sum_{\tau\in T_k(G_\Omega^{(\eta)})}W(\tau) \right) &= 0 \nonumber \\
\sum_{\substack{ \tau\in T_j(G_\Omega^{(\eta)})\\ k\in N_\eta ,k\ne j }} w_{kj}W(\tau) = 
  \sum_{\substack{ k\in N_\eta ,k\ne j \\ \tau\in T_k(G_\Omega^{(\eta)})}}w_{jk}W(\tau)& \label{cancel1} 
\end{align}
where we have discarded the denominator (which must be non-zero since there is at least one tree in $G_\Omega^{(\eta)}$). 

We have a sum of products of $w$-elements on each side, and there are $n_\eta:=\dim(\Omega^{(\eta)})$ elements in each product. The number of terms on each side is $(n_\eta-1)n_\eta^{n_\eta-2}$. That is, $n_\eta -1$ possible values of $k$ times the $n_\eta^{n_\eta-2}$ trees per root. We will show that there is a bijection between the products on the left and those on the right. The products on the left can be represented as $E_j(G_\Omega^{(\eta)})\times T_j({G_\Omega^{(\eta)}})$, where $E_j(\dot)$ represents the set of outgoing edges from node $j$. We will represent the products on the right as $\cup_{k\ne j} T_{k}(G_\Omega^{(\eta)})$. The products are a tree times the weight of edge $\overleftarrow{jk}$, but once the $k$-rooted tree is chosen, that edge is uniquely determined, therefore we only worry about the tree.

Let $r(\tau)$ represent the root of tree $\tau$, let $p(j,\tau)$ represent the parent vertex of node $j$ in a tree $\tau$, and let $d(j,k,\tau)$ represent the daughter of $j$ in a $j$-rooted tree $\tau$ in whose branch the vertex $k$ lies. The bijection $f_j$ that we are looking for, and its inverse $f^{-1}_j$, are as follows:
\begin{align*}
&f_j(\overleftarrow{kj},\tau) = \tau + \overleftarrow{kj} - \overleftarrow{j,d(j,k,\tau)} \\
&f_j^{-1}(\tau) = \left( \overleftarrow{p(j,\tau),j}, \tau - \overleftarrow{p(j,\tau),j} + \overleftarrow{j,r(\tau)} \right).
\end{align*}
$f_j$ is well-defined: the vertex $j$ is no longer a root because it now has an out-going edge, while $d(j,k,\tau)$  lost its outgoing edge so it is now the root. We now have a tree in $\cup_{k\ne j} T_{k}(G_\Omega^{(\eta)})$, since $d(j,k,\tau)$ cannot be the same of $j$. One may argue that $j$ may not have a daughter, but this can only happen if $N_\eta=\{j\}$, and this case is trivial (we have defined the weight of a trivial tree to be one, and our formula yields $\Lambda_{\eta,\infty}=e_j$, which must be the stationary state of a basin of one). We have not formed a cycle because we removed an edge from the cycle that was temporarily formed via the new edge $\overleftarrow{kj}$. 

$f_j^{-1}$ is well-defined. $\overleftarrow{p(j,\tau)}$ is in $E_j(G_\Omega^{(\eta)})$, by definition of $E_j()$. The old root $r(\tau)$ gains an outgoing edge and is no longer the root, while $j$ loses its outgoing edge, so it must be the new root. As before, we have not formed any new cycles, since the added and removed edges would be in the same cycle.

To show bijectivity: 
\begin{align*}
f_j \circ f_j^{-1}(\tau) =& \tau - \overleftarrow{p(j,\tau),j}+\overleftarrow{j,r(\tau)}  + \overleftarrow{p(j,\tau),j}\nonumber \\ &- \overleftarrow{j,d(j,k, \tau - \overleftarrow{p(j,\tau),j}+\overleftarrow{j,r(\tau)})}\\
&\hspace{-.3in}= \tau +\overleftarrow{j,r(\tau)} - \overleftarrow{j,d(j,k, \tau - \overleftarrow{p(j,\tau),j}+\overleftarrow{jk})}\\
&\hspace{-.3in}= \tau + \overleftarrow{j,r(\tau)} - \overleftarrow{j,r(\tau)} \\
&\hspace{-.3in}= \tau \\
f_j^{-1}\circ f_j (\overleftarrow{kj},\tau) =&  \left(\overleftarrow{p(j, f_j(\overleftarrow{kj}, \tau)),j}, \tau + \overleftarrow{kj} \nonumber  \right. \\
 & -\overleftarrow{j,d(j,k,\tau)} - \overleftarrow{p(j,f_j(\overleftarrow{kj},\tau)),j} \nonumber \\
& \left. +\overleftarrow{j,r(f_j(\overleftarrow{kj},\tau))} \right)  \\
= &(\overleftarrow{kj}, \tau + \overleftarrow{kj}  - \overleftarrow{j,d(j,k,\tau)}) \nonumber \\
& - \overleftarrow{kj}+\overleftarrow{j, d(j,k,\tau)}) )\\
=&(\overleftarrow{kj}, \tau ). 
\end{align*}
So we can map every product $w_{kj}W(\tau)$ on the left-hand side of (\ref{cancel1}) to a term $w_{jk}W(\tau)$ on the right-hand side, and vice versa. It follows that everything cancels, and the equation is satisfied. Thus $\Lambda_{\eta,\infty}$ is in the kernel of $\mathcal{T}$.
\end{proof}

\section{Constraints}

We have found an interpretation of the kernel of $\Omega$ in terms of trees in $G_\Omega$. It is also possible to describe the kernel of $\Omega^T$ in terms of forests. Since the kernel of both has dimension $n_B$, there should be $n_B$ independent vectors $\kappa_\eta \in \ker(\Omega^T)$ such that $\kappa_\eta^T\Lambda$ is instantaneously stationary. And if $\Omega$ is time-independent, these quantities should be conserved.

Define the shorthand $T_B(G_\Omega):=T_{N_B(G_\Omega)}(G_\Omega)$. Also define $T_B(G_\Omega, \eta, l)$ to the set of forests in $T_B(G_\Omega)$ such that node $l$ is in the same tree as basin $\eta$. Then we have the following theorem:

\begin{theorem}
The kernel of $\Omega^T$ is spanned by the following linearly independent vectors:
\begin{align*}
\kappa_\eta := \left(\sum_{\tau\in T_{B}(G_\Omega)}W(\tau)\right) \sum_{j\in N_\eta } e_j  + 
\sum_{\substack{l \not\in N_B(G_\Omega) \\ \tau\in T_B(G_\Omega,\eta,l)}} W(\tau) e_l.
\end{align*}
If $\Omega$ has no time dependence, the quantities $\kappa_\eta^T \Lambda$ are conserved.
\end{theorem}
\begin{proof}
The linear independence of these vectors is fairly clear. Each $\kappa_\eta$ has non-zero components in $N_\eta$, but zero components in the other basins $N_{\eta'}$, $\eta'\ne \eta$. Since we know the dimension of the kernel is $n_B$, and we have $n_B$ linearly dependenet vectors, we just need to show that an arbitrary row of $\Omega^T$ multiplied by $\kappa_\eta$ gives zero for any $\eta$. There are two cases: rows corresponding to nodes inside and outside the basins.

In the first case, the non-zero elements are in the same basin (otherwise, there would be an edge leading out of the basin). For arbitrary $j\in N_\eta$, we have:
\begin{align*}
e_j^T \Omega^T \kappa_{\eta'} &= \delta_{\eta\eta'} \sum_{\substack{k\in N_\eta \\ \tau\in T_B(G_\Omega)}} W(\tau) e_j^T \Omega^T e_k \\ 
&= \delta_{\eta\eta'} \sum_{\tau\in T_B(G_\Omega)} W(\tau) \left(\sum_{k\in N_\eta} e_k^T \Omega  e_j \right)\\
&= 0,
\end{align*}
where $\sum_{k\in N_\eta} e_k^T \Omega  e_j = 0$ due to the vanishing column-sums of $\Omega$, and the fact that the only non-zero elements are in $N_\eta$.

In the second case, we have an arbitrary row $j$ outside of the basins. This gives:
\begin{align*}
e_j^T \Omega^T \kappa_{\eta} = \sum_{\substack{k\in N_\eta \\ \tau\in T_B(G_\Omega)}} W(\tau) e_j^T \Omega^T e_k + \\
\sum_{\substack{k \not\in N_B(G_\Omega) \\ \tau\in T_B(G_\Omega,\eta,k)}} W(\tau) e_j^T \Omega^T e_k. \\
\end{align*}
Unlike in the previous case, there is no vanishing column-sum, as the first sum does not cover all of the non-zero $\Omega$-elements, nor does the second (and the coefficients $W(\tau)$ depend on $k$). If we expand the sums so that we have a sum of products of $w$-elements, and set the expression to zero, we get:
\begin{align}
\sum_{\substack{k\in N_\eta \\ \tau\in T_B(G_\Omega)}} w_{kj} W(\tau) + \sum_{\substack{k \not\in N_B(G_\Omega) \\ k\ne j \\ \tau\in T_B(G_\Omega,\eta,k)}} w_{kj} W(\tau)\nonumber \\ = \sum_{\substack{k \ne j \\ \tau\in T_B(G_\Omega,\eta,j)}} w_{kj} W(\tau) \label{cancel2}
\end{align}
To show this equation is true, our strategy will be the same as the previous theorem: we show that there are the same number of terms on each side, and there is a bijection between them.

Let $n_F:= |T_B(G_\Omega)|$ and $n_{F,\eta}:=|T_B(G_\Omega,\eta,j)|$. Note that the latter is independent of $j$, and $n_F = \sum_{\eta=1}^{n_B}n_{F,\eta}$. Furthermore, we have
\begin{align*}
\frac{n_{F,\eta}}{n_F} = \frac{|N_\eta|}{|N_B(G_\Omega)|}.
\end{align*}
This is true because the number of forests with node $j$ in the tree rooted in basin $\eta$ must be proportional to the number of vertices in the basin (if $j$'s path to the root reaches the basin first at node $k$, there are $|N_\eta|-1$ other forests, where we reattach $j$'s path to the other nodes in the basin instead of $k$). 

So to count the number of terms in (\ref{cancel2}), we have $|N_\eta|n_F$ in the first sum of the LHS, and $(n-|N_B(G_\Omega)|-1)n_{F,\eta}$ in the second sum. On the RHS, we have $(n-1)n_{F,\eta}$ terms. This combines to:
\begin{align*}
|N_\eta|n_F + (n-|N_B(G_\Omega)|-1)n_{F,\eta} &= (n-1)n_{F,\eta} \\
|N_\eta|n_F &= |N_B(G_\Omega)| n_{F,\eta} \\
|N_\eta|n_F &= |N_\eta|n_F,
\end{align*}
where we have used the ratio of forest numbers above in the last line.

So we have the same number of terms on each sides, and each term has the same number of $w$-factors. We must find a bijection between the terms to show cancellation. We can represent the set on the LHS as $S_{L,1}\cup S_{L,2}$ and the set on the right as $S_R$, where:
\begin{align*}
S_{L,1}&:= \{(\overleftarrow{kj},\tau): k\in N_\eta, \tau\in T_B(G_\Omega) \} \\
S_{L,2}&:= \{(\overleftarrow{kj},\tau): k\not\in N_B(G_\Omega), k\ne j, \tau\in T_B(G_\Omega, \eta, k) \} \\
S_{R}&:= \{(\overleftarrow{kj},\tau): k\ne j, \tau\in T_B(G_\Omega, \eta, j) \}.
\end{align*}
Define our bijection $f:S_{L,1}\cup S_{L,2}\rightarrow S_R$ as follows:
\begin{enumerate}
\item If $k$ is upstream from $j$, then $f(\overleftarrow{kj},\tau) := (\overleftarrow{kj},\tau)$.
\item Otherwise, $f(\overleftarrow{kj},\tau) := (\overleftarrow{p(j,\tau),j}, \tau + \overleftarrow{kj} - \overleftarrow{p(j,\tau),j} )$.
\end{enumerate}
Note that if we write $(e',\tau')=f(e,\tau)$, then $\tau'$ must be a forest if $\tau$ is: node $j$ in $\tau'$ has one outgoing node, and no cycles have been formed, since $j$ is not attached to an upstream node. 

We can also define $f^{-1}$ in the same way as $f$.  Because $f$ and $f^{-1}$ have different domains, they are not the same function, but both require swapping the input edge with $j$'s outgoing edge provided the input edge is not upstream. We must still show that $f$ and $f^{-1}$ have $S_R$ and $S_{L,1}\cup S_{L,2}$ as images.

First, $f(S_{L,1})\subseteq S_R$. The edge $e$ points to a node in $N_\eta$, so it can't be upstream from $j$, which is not in a basin. So $e$ is swapped in, and since it points to $N_\eta$, $j$ is now in the same tree. Therefore $\tau' \in T_B(G_\Omega,\eta,j)$. And $e'$ does not point to $j$ since it is outgoing from $j$, so $(e',\tau')\in S_R$.

Next, $f(S_{L,2})\subseteq S_R$. If $k$ is upstream from $j$, then $\tau$ must be in $T_B(G_\Omega,\eta,j)$, since it is also in $T_B(G_\Omega, \eta, k)$, and $j$ and $k$ are in the same tree. No swap is made, $\tau'=\tau$, and so $\tau'$ is in $T_B(G_\Omega,\eta,j)$. Therefore $(e',\tau')\in S_R$. If $k$ is not upstream from $j$, the edges are swapped. Since $\tau\in T_B(G_\Omega,\eta,k)$, attaching $j$ to $k$ puts $\tau'\in T_B(G_\Omega,\eta,j)$, so $(e',\tau')\in S_R$. Therefore we have $f(S_{L,1}\cup S_{L,2}) \subseteq S_R$.

Finally, $f^{-1}(S_R)\subseteq S_{L,1}\cup S_{L,2}$. Here, we identify $f^{-1}(e',\tau') = (e, \tau)$. There are three cases: (1) $e'$ pointing from $j$ to a node in $N_\eta$ and (2) $e'$ pointing from $j$ to a node outside of $N_B(G_\Omega)$ that is not upstream from $j$, and (3) $e'$ pointing from $j$ to a node upstream. In the first case, we remove $e'$ and replace with $e$ (since $k$ cannot be upstream if it in the basin), and since $e'$ points to $N_\eta$, and $\tau$ is some forest in $T_B(G_\Omega)$, then $(e,\tau)\subseteq S_{L,1}$. In the second case, the edges are swapped. Since $\tau'\in T_B(G_\Omega,\eta,j)$, it is also in $T_B(G_\Omega,\eta,p(j,\tau'))$ ($j$ must have a parent, else it would form its own basin). So we have $f^{-1}(\overleftarrow{kj}, \tau') = (\overleftarrow{p(j,\tau'),j}, \tau)$. If $p(j,\tau')$ is in $N_\eta$, then the edge and tree must be in $S_{L,1}$. If it is a non-basin node, then $\tau\in T_B(G_\Omega,\eta,p(j,\tau'))$ and the edge and tree must be in $S_{L,2}$. In the third and final case, no swap is made. Since $k$ is upstream, and $\tau'\in T_B(G_\Omega,\eta,j)$, then $\tau'$ is also in $T_B(G_\Omega,\eta,k)$. $k$ must also be a non-basin node, so $(e',\tau')=(e,\tau)\in S_{L,2}$. 

We can conclude that $f$ is a bijection, and therefore (\ref{cancel2}) is true.
\end{proof}

From a computational standpoint, it is easier to calculate a parallel set of $\kappa'_\eta$-vectors that factor out the piece inside the basins. Define $G_\Omega'$ to be the graph $G_\Omega$ with all edges inside the basins removed (but not edges from outside the basins to the basins). $T_B(G_\Omega')$ is defined in the same way, and $T_B(G_\Omega',\eta,l)$ is the set of all forests with roots in the basins such that node $l$ is in a tree rooted in a node corresponding to basin $\eta$. Then we have: 
\begin{corollary}
The kernel of $\Omega^T$ is spanned by the following linearly independent vectors:
\begin{align*}
\kappa_\eta' := \left(\sum_{\tau\in T_{B}(G_\Omega')}W(\tau)\right) \sum_{j\in N_\eta } e_j  + 
\sum_{\substack{l \not\in N_B(G_\Omega) \\ \tau\in T_B(G_\Omega',\eta,l)}} W(\tau) e_l.
\end{align*}
\end{corollary}
\begin{proof}
The proof proceeds in the same way as the theorem.
\end{proof}
The forests in these sums will have fewer branches, and there will be fewer of them, so the calculations are quicker.

If we have a set of stationary projectors, we can uniquely identify the asymptotic state, provided the initial eigenprojectors coincide:
\begin{corollary}
Suppose a density operator has initial state $\rho = \sum_{j=1}^n \Lambda_{0,j} \pi_j$, where $\{\pi_j\}$ are mutually orthogonal rank-one projectors that commute with the Hamiltonian and Lindblad operators. Then the system approaches the following state asymptotically:
\begin{align*}
\rho_\infty &:=\sum_{\eta=1}^{n_B}\sum_{\substack{j\in N_\eta \\ \tau \in T_j(G_\Omega)}} c_\eta W(\tau)\pi_j \\
c_\eta &:= \sum_{j\in N_\eta} \Lambda_{0,j} + \frac{\sum_{l\not\in N_B(G_\Omega)} \sum_{\tau\in T_B(G_\Omega',\eta,l)} W(\tau)\Lambda_{0,l}}{\sum_{\tau\in T_B(G_\Omega')}W(\tau)}
\end{align*}
\end{corollary}
\begin{proof}
We set $\kappa_\eta'^T \left(\sum_{\eta'=1}^{n_B}c_\eta \Lambda_{\eta,\infty}\right) = \kappa_\eta'^T\Lambda_0$. This yields:
\begin{align*}
&\left(\sum_{\tau\in T_{B}(G_\Omega')}W(\tau)\right) c_\eta\sum_{j\in N_\eta}\Lambda_{\eta,\infty, j} = \\
&\left(\sum_{\tau\in T_{B}(G_\Omega')}W(\tau)\right) \sum_{j\in N_\eta} \Lambda_{0,j}   +\sum_{\substack{l \not\in N_B(G_\Omega) \\ \tau\in T_B(G_\Omega',\eta,l)}} W(\tau) \Lambda_{0,l}.  \\
&c_\eta = \sum_{j\in N_\eta} \Lambda_{0,j} + \frac{\sum_{\substack{l \not\in N_B(G_\Omega) \\ \tau\in T_B(G_\Omega',\eta,l)}} W(\tau) \Lambda_{0,l}}{\sum_{\tau\in T_{B}(G_\Omega')}W(\tau)}
\end{align*}
\end{proof}

\section{Asymptotic Behavior of GPM systems}

The final corollary of the preceding section might be considered of limited use, since it requires the system to start with stationary projectors. An arbitrary initial state will not have stationary projectors, and, in general, it may difficult to determine what the stationary projectors are. In this section, we look at the stationary projectors for GPM-systems, since they are often used to model Lindblad processes.

We start with the ``generic" case: 
\begin{proposition}
Suppose we have a diagonal Hamiltonian $H$, and $N$ Lindblad operators $L_\alpha=A_\alpha D_\alpha$, where $A_\alpha$ are permutation matrices and $D_\alpha$ are diagonal. Form the graph $G_\Omega$ corresponding to the diagonal projectors $|j\rangle\langle j|$, and define $P_\eta:= \sum_{j\in N_\eta}|j\rangle\langle j|$ and $P_B:=\sum_{\eta}^{n_B} P_\eta$. Then the decay space is the subspace corresponding to $I-P_B$. Moreover, $P_\eta$ commutes with $P_BHP_B$ and $P_BL_\alpha P_B$ for any $\alpha$.
\end{proposition}
\begin{proof}
We have shown in subsection IIB that the projectors $|j\rangle\langle j|$ are stationary. Therefore if $\rho(0)$ is an arbitrary superposition of such projectors, by Theorem III.3., the system will approach a superposition of $\sum_{j\in N_\eta}\Lambda_{\eta,\infty,j}|j\rangle\langle j|$. Each of these component states are full-rank on the images of the corresponding $P_\eta$. Since there is no asymptotic piece in the image of $I-P_B$, we can conclude that this space is inside the decay space. Since there are asymptotic states that are full-rank in the image of $P_B$, we can conclude the decay space is no bigger than the image of $I-P_B$.

Moreover, $P_BL_\alpha P_B$ is block-diagonal with respect to $P_\eta$: each block is a basin, and therefore there cannot be edges between them. The non-basin sector has been projected away, so there are no in-going edges to any of the basins. Hence $[P_\eta, P_BL_\alpha P_B]=0$. 

$H$, $P_\eta$, and $P_B$ are all diagonal, so $[P_\eta, P_BH P_B]=0$, and we are finished.
\end{proof}

The question arises: are these the only stationary projectors? The answer is that usually they are, but there are cases where symmetries ``hide" other enclosures (and possibly oscillations between enclosures). There are three cases: (1) when $P_\eta$ can be written as a sum of smaller commuting projectors, in which case there are smaller enclosures, (2) when there is a continuous transformation between two (or more) different $P_\eta$ where each transient projector commutes with the Lindblad and Hamiltonian operators, in which case there is a continuum of enclosures, and (3) there is a continuum of projectors that commute with $P_BL_\alpha P_B$, and an altered Hamiltonian $P_BHP_B-h_1P_{\eta_1} - h_2P_{\eta_2}$, in which case there are oscillations. In the first two cases, we will refer to spaces corresponding to commuting projectors other than $P_\eta$ as \emph{hidden enclosures}. 

Before stating the conditions for hidden enclosures, we define some objects. The symmetry will be described by an equivalence relation $\sim$ on a subset $N_\sim$ of $\mathbb{Z}_n$. This relation will be \emph{uniform}, by which we mean that all equivalence classes will have the same number of elements. $\mu_\sim$ will be the number of equivalence classes, and $\nu_\sim$ will be the number of elements in each class, so that $|N_\sim| = \mu_\sim\nu_\sim$.

Besides the symmetry, hidden enclosures also require a type of resonance among the Lindblad operators to exist. To describe this resonance, we define the \emph{induced coherence graph} $G_\sim$:
\begin{itemize}
\item The vertices are double-indexed and represent the coherences of equivalent states: $N_{C}:=\mathcal{V}(G_\sim)=\{(j,k): j,k\in N_\sim; j\sim k, j\ne k\}$ 
\item Edges are directed; self- and multi-edges are allowed: $\mathcal{E}(G_\sim):=\{(j_1,k_1)\rightarrow_\alpha (j_2,k_2): \exists \alpha, \sigma_\alpha(j_1) = j_2, \sigma_\alpha(k_1) = k_2, D_{\alpha,j_1}, D_{\alpha,k_1}\ne 0\}$.
\item The weights are $U(1)$-valued: $w(j_1,k_1,j_2,k_2,\alpha):= \frac{D_{\alpha,j_1}D_{\alpha,k_1}^*}{|D_{\alpha,j_1}||D_{\alpha,k_1}|}$. 
\end{itemize}
An induced coherence graph is \emph{resonant} if and only if there exists a function $f_C: N_C\times N_C \rightarrow U(1)$ such that:
\begin{enumerate}
\item $f_C$ is transitive in the sense that $f_C((j_1,k_1),(j_2,k_2))f_C((j_2,k_2),(j_3,k_3)) = f_C((j_1,k_1),(j_3,k_3))$ for any triplet of vertices in $G_\sim$. This also implies that $f_C((j,k),(j,k))=1$ and $f_C((j_1,k_1),(j_2,k_2)) = f_C((j_2,k_2),(j_1,k_1))^*$.
\item The function matches the weights for any edge in $G_\sim$: $f_C((j_1,k_1),(j_2,k_2)) = w(j_1,k_1,j_2,k_2,\alpha)$, regardless of $\alpha$.
\end{enumerate}
For example, any self-edge must have weight one, which means that if $L_\alpha$ has two or more $1$-cycles whose elements are equivalent, their phases must match. This is an example of inter-cycle resonance. Let us define a cycle-averaged phase as:
\begin{align*}
\bar{\theta}_{\alpha,\gamma} := \frac{1}{m_{\alpha,\gamma}}\sum_{j\in N_{\alpha,\gamma}} \arg(D_{\alpha,j}).
\end{align*}
If we have two cycles in a Lindblad operator such that $N_{\alpha,\gamma_1}, N_{\alpha,\gamma_2} \subset N_\sim$, and the equivalence classes straddle these cycles (\emph{i.e.} each node in the first cycle is equivalent to some node in the second), then resonance of $G_\sim$ implies that:
\begin{align*}
\bar{\theta}_{\alpha,\gamma_1} = \bar{\theta}_{\alpha,\gamma_2} \mod \left(\frac{2\pi}{\lcm(m_{\alpha,\gamma_1}, m_{\alpha,\gamma_2})}\right)
\end{align*}

Here is a lemma we will require:
\begin{lemma}
The weakly connected components of $G_\sim$ are strongly connected. That is, if there is a path from one node to another, there must be a path in the opposite direction.
\end{lemma}
\begin{proof}
Consider a path between nodes $(j_1,k_1)$ and $(j_2,k_2)$. We can represent it as a product $\tau$ of $n_\tau$ permutations $\tau=\overleftarrow\Pi_{j=1}^{n_\tau}\sigma_{\alpha_j}$, where the arrow indicates the product order is the reverse of the conventional order. We have $j_2=\tau(j_1)$ and $k_2=\tau(k_1)$. For the path to exist, the $D$-elements along the way must be non-zero, \emph{i.e.} $D_{\alpha_j, \tilde{j}_j}, D_{\alpha_j, \tilde{k}_j}\ne 0$, where $\tilde{j}_j = \sigma_{\alpha_j}\circ\cdots\circ\sigma_{\alpha_1}(j_1)$ and similarly for $\tilde{k}_j$. 

Since all nodes $j$ and $k$ are in basins, which are strongly connected, there is a path $\tau'$ that leads from $j_2$ back to $j_1$: $\tau'(j_2)=j_1$. Similar to $\tau$, the  $D$-elements must be non-zero en route. Now, $\tau'(k_2)$ may not equal $k_1$, so we cannot use $\tau'$ to generate the reverse path in $G^\sim$. However, $\tau'$ will generate a path from $k_2$ to some other node  that is in the same equivalence class as $j_2$ and $k_2$. This is because, as we apply permutations iteratively to $k_2$, the $D$-elements along the way are non-zero, so there is an edge for each permutation. The reason we know the elements are non-zero is because each node en route is equivalent to the corresponding node in the path from $j_2$ to $j_1$. Since we know the $D$-elements in that path are non-zero, and the $D$-elements in $k_2$'s are equivalent, the symmetry assures the latter are non-zero.

So, $\tau'$ cannot lead us back to $k_1$ directly, but if we consider products of the form $\tau'(\tau\tau')^l$, we will get there. Note that $\tau'\circ(\tau\circ\tau')^l(j_2) = j_1$, for any exponent $l$. Moreover, the products $\tau'\circ(\tau\circ\tau')^l(k_2)$ permute $k_2$ cyclically, and because $\tau'\circ(\tau\circ\tau')^{-1}(k_2)=k_1$, $k_1$ must be in that cycle. If the cycle has length $n'$, then we must have $\tau'\circ(\tau\circ\tau')^{n'-1}(k_2)=k_1$. Does this product represent a path in $G_\sim$? Yes, because $\tau$ applied to $k_1$ and $\tau'$ applied to $k_2$ represent paths, and if they are applied to nodes equivalent to these nodes, the symmetry assures the $D$-elements are non-zero. These products respect the equivalence structure, so we have a path.

In summary, since $\tau'\circ(\tau\circ\tau')^{n'-1}$ maps $j_2$ to $j_1$, and $k_2$ to $k_1$, and all $D$-elements en route are non-zero, we have constructed a path in $G_\sim$ from $(j_2,k_2)$ to $(j_1, k_1)$. 
\end{proof}
In the theorem, we refer to connected components of $G_\sim$; because strong and weak connectedness are equivalent, we do not need to qualify the type of connectedness.

Now we can state necessary and sufficient conditions for hidden enclosures:
\begin{theorem}
There is a hidden enclosure in a GPM-system if and only if there is a uniform equivalence relation $\sim$ with $\nu_\sim \ge 2$ such that the following properties hold:
\begin{description}
\item[Hamiltonian symmetry] $j\sim k \implies H_j = H_k$.
\item[Dissipation symmetry] for any $\alpha$, $j\sim k \implies |D_{\alpha,j}| = |D_{\alpha,k}|$. Additionally, $j\sim k \implies \sigma_\alpha(j) = \sigma_\alpha(k)$ whenever $\sigma_\alpha(j), \sigma_\alpha(k) \in N_B$.
\item[Dissipation resonance] The induced coherence graph $G_\sim$ has a connected component that is resonant.
\end{description}
\end{theorem}
\begin{proof}
We first show that the existence of $\sim$ is necessary. Let $P$ be any projector that commutes with $P_BHP_B$ and $P_BL_\alpha P_B$ and is neither a basin projector $P_\eta$ nor a sum thereof (but whose image is contained in the image of $P_B$). $P$ must have at least one off-diagonal element. If $P$ were diagonal, we would have: 
\begin{align*}
\sum_{j\in N_B} [P_{jj}|j\rangle\langle j|, L_\alpha] &= 
\sum_{j\in N_B} P_{jj}\left(D_{\alpha,\sigma_\alpha^{-1}(j)} |j\rangle\langle \sigma^{-1}_\alpha(j)| \right.  \\ & \hspace{1in} - D_{\alpha,j}|\sigma_\alpha(j)\rangle\langle j|\Big)   \\
&\hspace{-0.5in}= \sum_{j\in N_B} D_{\alpha,j}|\sigma_\alpha(j)\rangle\langle j|\left( P_{\sigma_\alpha(j)\sigma_\alpha(j)} - P_{jj}\right). 
\end{align*}
For this to vanish, we require $P_{\sigma_\alpha(j)\sigma_\alpha(j)} = P_{jj}$ as long as $D_{\alpha,j}\ne 0$. Since a basin is strongly connected, all diagonal elements would be the same. Diagonal projectors can only have ones or zeros on the diagonal, which would give us a $P_\eta$ or a sum thereof, which contradicts our assumption. Therefore there is  at least one off-diagonal element $P_{j'k'}$.

Using similar logic as above, there is a recurrence relation for the off-diagonal elements:
\begin{align}
P_{\sigma_{\alpha}(j)\sigma_{\alpha}(k)} D_{\alpha,k}  = P_{jk} D_{\alpha,j}. \label{recur1}
\end{align}
Additionally, using the Hermitian property of projectors:
\begin{align}
P_{\sigma_{\alpha}(j)\sigma_{\alpha}(k)} D_{\alpha,j}^* &= (P_{\sigma_{\alpha}(k)\sigma_{\alpha}(j)} D_{\alpha,j})^* = (P_{kj} D_{\alpha,k})^*\nonumber \\ 
&= P_{jk} D_{\alpha,k}^* \label{recur2}
\end{align}
Using (\ref{recur1}) and (\ref{recur2}), we get:
\begin{align}
P_{jk}|D_{\alpha,j}|^2 &= P_{\sigma_{\alpha}(j)\sigma_{\alpha}(k)} D_{\alpha,k}D_{\alpha,j}^* \nonumber \\ &= P_{jk} D_{\alpha,k}^*D_{\alpha,k} \nonumber \\ 
P_{jk}|D_{\alpha,j}|^2 &= P_{jk}|D_{\alpha,k}|^2. \label{recur3}
\end{align}
Equation (\ref{recur3}) says that if $P_{jk}\ne 0$, then $|D_{\alpha,j}|=|D_{\alpha,k}|$. If $j'\in N_{\eta_j}$ and $k'\in N_{\eta_k}$, we define a binary relation $\approx$ on $N_\sim:=N_{\eta_j}\cup N_{\eta_k}$ by saying $j\approx k$ if and only if $P_{jk}\ne 0$. This relation obeys the symmetry property of binary relations due to the Hermitian property of projectors. It is also reflexive. Due to the strong connectivity combined with (\ref{recur1}), there is some pair $(j,k)$ for any $j\in N_\sim$ such that $j\approx k$. Either $k=j$ and therefore $j\approx j$, or $k\ne j$, which means that $P_{jk}\ne 0$, which means that $P_{jj}\ne 0$ and $j\approx k$, since diagonal entries of projectors dominate off-diagonal entries $|P_{jk}|^2 \le P_{jj}P_{kk}$. The binary relation may not be transitive, so it may not be an equivalence relation. However, we can easily extend the relation: define the equivalence relation $\sim$ such that $j\sim k$ if and only if there is a sequence $j \approx l_1 \approx l_2 \cdots \approx k$. It is clear that, due to the transitivity of equality, we have $j\sim k \implies |D_{\alpha,j}| = |D_{\alpha,k}|$. 

We now prove the preservation of equivalence classes (within basins) by the permutation $\sigma_\alpha$. Because of (\ref{recur1}), $j\sim k \implies \sigma_\alpha(j)\sim \sigma_\alpha(k)$ if $D_{\alpha,j},D_{\alpha,k}\ne 0$. It is possible however that the respective $D$-elements are zero. In this case, we can iterate the permutation backward. Relations (\ref{recur1}) and (\ref{recur2}) can be written $P_{jk}D_{\sigma_\alpha^{-1}(k)} = P_{\sigma_\alpha^{-1}(j)\sigma_\alpha^{-1}(k)}D_{\sigma_\alpha^{-1}(j)}$, and $P_{jk}D_{\sigma_\alpha^{-1}(j)}^* = P_{\sigma_\alpha^{-1}(j)\sigma_\alpha^{-1}(k)}D_{\sigma_\alpha^{-1}(k)}^*$. This implies that if $j \approx k$, then $\sigma_{\alpha}^{-1}(j)\approx\sigma_{\alpha}^{-1}(k)$ unless both $D_{\sigma_\alpha^{-1}}(j)$ and $D_{\sigma_\alpha^{-1}}(k)$ are zero. But due to Proposition \ref{GPMzeros}, we know that there can only be one zero $D$-element per cycle. Therefore if we go backward, we can't hit zero $D$-elements until we get to our starting point. It follows that the relation $\approx$ holds for each pair in the cycle, and also that both cycles must be the same length. Since the binary relation $\approx$ is preserved by permutations within basins, the extension $\sim$ must also be preserved.

The uniformity of $\sim$ follows. Since permutations are bijections, equivalence classes that are connected by permutations must have the same size. Since we have at most two basins, and our base equivalence class $[j'k']$ has elements in those basins (by definition), every element in $N_\sim$ is in an equivalence class that has the same size as $[j'k']$. And since the base equivalence class has at least two elements, $\nu_\sim \ge 2$. Note that while $\mu_\sim$ clearly divides $|N_\sim|$, it also must divide $|N_{\eta_j}|$ and $|N_{\eta_k}|$. 

This proves the necessary existence of a uniform equivalence relation with dissipation symmetry. The Hamiltonian symmetry is relatively straightforward. We require that $P$ commute with $P_BHP_B$:
\begin{align*}
[P_BHP_B,P] &= \sum_{j,k\in N_\sim} \left(H_j - H_k \right) P_{jk}|j\rangle\langle k|.
\end{align*}
This means that $j\approx k\implies H_j=H_k$. This implies in turn that $j\sim k\implies H_j = H_k$.

We now turn to resonance. To show the resonance of $G_\sim$ is necessary and sufficient, we will look at the ODE governing the coherences $\rho_{jk}$, $j\ne k$, $j\sim k$:
\begin{align*}
\dot{\rho}_{jk} &= \sum_{\alpha}\left(D_{\alpha,\sigma_\alpha^{-1}(j)}D_{\alpha,\sigma_\alpha^{-1}(k)}^* \rho_{\sigma_\alpha^{-1}(j)\sigma_\alpha^{-1}(k)} \right. \\
& \hspace{0.5in} + \left. \left(iH_k-iH_{j}-\frac{|D_{\alpha,j}|^2+|D_{\alpha,k}|^2}{2}\right)\rho_{jk} \right) \\
&= \sum_{\alpha}  \left( |D_{\alpha,\sigma_\alpha^{-1}(j)}|^2 w(\sigma_\alpha^{-1}(j),\sigma_\alpha^{-1}(k), j, k, \alpha)\times \right. \\
& \hspace{1in} \left. \rho_{\sigma_\alpha^{-1}(j)\sigma_\alpha^{-1}(k)} - |D_{\alpha,j}|^2 \rho_{jk} \right), 
\end{align*} 
where we have used the symmetries $H_j=H_k$ and $|D_{\alpha,j}| = |D_{\alpha,k}|$ since $j\sim k$. This is an ODE of dimension equal to $|N_C|$, the order of $G_\sim$. We then define matrices $A^\sim_\alpha$ and $A^\sim$ where the rows and columns are labelled according to the vertices of $G_\sim$, which has the following elements:
\begin{align*}
(A^\sim_\alpha)_{(j_1,k_1),(j_2,k_2)} &= |D_{\alpha,j_2}|^2 \times \\ 
&\left(\delta_{j_1,\sigma_\alpha(j_2)}\delta_{k_1,\sigma_\alpha(k_2)} w(j_2,j_1,k_2,k_1,\alpha) \right.\\
& \hspace{1in}- \delta_{j_1j_2}\delta_{k_1k_2}  \Big) \\
(A^\sim)_{(j_1,k_1),(j_2,k_2)} &= \sum_{\alpha=1}^N (A^\sim_\alpha)_{(j_1,k_1),(j_2,k_2)}.
\end{align*}
Note that we are ignoring coherences for $j\not\sim k$. If a coherence $\rho_{jk}$ where $j\not\sim k$ were to have a non-zero stationary solution, the projector element $P_{jk}$ would have to be non-zero, which would contradict the definition of $\sim$. 

Since hidden enclosures imply stationary off-diagonal elements of $\rho$ (and vice versa), the existence of hidden enclosures is equivalent to $A^\sim$ being rank-deficient and having a non-empty kernel. One may object that the elements of the kernel could correspond to non-Hermitian $\rho$, but because $(2LAL^\dagger -L^\dagger LA - AL^\dagger L)^\dagger = 2LA^\dagger L^\dagger -L^\dagger LA^\dagger - A^\dagger L^\dagger L$ for any operator $A$, any element in the kernel implies the existence of a stationary Hermitian $\rho$ with off-diagonal nonzero elements.

We will show that the resonance of a component of $G_\sim$ is equivalent to a non-empty kernel of $A^{\sim \dagger}$. The kernels of a matrix and its conjugate of course are not equivalent, but they must have the same dimension. The reason for working with the conjugate will become clear shortly.

First we note that $A^{\sim \dagger}$ is block-diagonal with blocks corresponding to the connected components of $G_\sim$. For this reason, we only need to consider vectors with support on individual blocks. Let $v$ be such a vector, and let $N_v\subseteq N_C$ be the indices of the relevant block. Choose a node $(\bar{\jmath},\bar{k}) \in N_v$ such that $|v_{(\bar{\jmath},\bar{k})}| = \sup_{(j,k)\in N_v}\{|v_{(j,k)}|\}$. Without loss of generality, we can set $v_{(\bar{\jmath},\bar{k})} = 1$.

Now let's look at what happens when the $(\bar{\jmath},\bar{k})$'th row of $A^{\sim\dagger}$ is multiplied by $v$:
\begin{align*}
\sum_{ (j,k)\in N_v} A^{\sim *}_{(j,k),(\bar{\jmath},\bar{k})}v_{(j,k)} &= \sum_{\alpha=1}^N | D_{\alpha, \bar{\jmath}} |^2 \times \\
&\hspace{-.8in}\left(w(\bar{\jmath},\bar{k},\sigma_\alpha(\bar{\jmath}),\sigma_\alpha(\bar{k}),\alpha)^* v_{\sigma_\alpha(\bar{\jmath}),\sigma_\alpha(\bar{k})}  -1 \right)
\end{align*}
Because $|v_{(j,k)}|\le 1$, for this equation to be satisfied, we have, for all $\alpha$ such that $D_{\alpha,\bar{\jmath}} \ne 0$:
\begin{align*}
w(\bar{\jmath},\bar{k},\sigma_\alpha(\bar{\jmath}),\sigma_\alpha(\bar{k}),\alpha)^* v_{\sigma_\alpha(\bar{\jmath}),\sigma_\alpha(\bar{k})}  =1,
\end{align*}
which implies that $|v_{\sigma_\alpha(\bar{\jmath}),\sigma_\alpha(\bar{k})}| = 1$ and $\arg(v_{\sigma_\alpha(\bar{\jmath}),\sigma_\alpha(\bar{k})}) = w(\bar{\jmath},\bar{k},\sigma_\alpha(\bar{\jmath}),\sigma_\alpha(\bar{k}),\alpha)$. 

Because of the strong connectedness of the component of $G_\sim$, this process can clearly be iterated to show that all elements of $v$ corresponding to $N_v$ have magnitude one. Moreover, for any path $(j_0,k_0)=(\bar{\jmath},\bar{k})\rightarrow_{\alpha_1} (j_1,k_1) \cdots \rightarrow_{\alpha_{n_p}}(\tilde{\jmath},\tilde{k})=(j_{n_p},k_{n_p})$, we have:
\begin{align*}
v_{(\tilde{\jmath},\tilde{k})} &= \prod_{l=1}^{n_p} w(j_{l-1}, k_{l-1}, j_l, k_l, \alpha_l).
\end{align*}
This equation essentially proves our theorem. If there is a function $f_C$ that matches the edge weights, regardless of $\alpha$, we can set
\begin{align*}
v_{(\tilde{\jmath},\tilde{k})} &= f_C((\bar{\jmath},\bar{k}), (\tilde{\jmath},\tilde{k})) \\
&= \prod_{l=1}^{n_p} f_C( (j_{l-1}, k_{l-1}), (j_l, k_l)) \\
&= \prod_{l=1}^{n_p} w(j_{l-1}, k_{l-1}, j_l, k_l, \alpha_l),
\end{align*}
and we have constructed a vector in the kernel of $A^\sim$. 

On the other hand, if there is no such function, the kernel must be empty. If we define a function $f_v$ on $N_v\times N_v$:
\begin{align*}
f_v( (j_1,k_1), (j_2,k_2) )&:= v_{(j_1,k_1)}^*v_{(j_2,k_2)},
\end{align*}
which must obey transitivity:
\begin{align*}
&f_v( (j_1,k_1), (j_2,k_2) ) f_v( (j_2,k_2), (j_3,k_3) ) \\
&\hspace{1.2in}= v_{(j_1,k_1)}^*v_{(j_2,k_2)}v_{(j_2,k_2)}^*v_{(j_3,k_3)} \\
&\hspace{1.2in}= v_{(j_1,k_1)}^*v_{(j_3,k_3)} \\
&\hspace{1.2in}= f_v( (j_1,k_1), (j_3,k_3) ).
\end{align*}
So a non-empty kernel guarantees the existence of the required function.
\end{proof}

\begin{remark}
We mentioned previously that there were two cases of hidden enclosures: (1) those contained in (and smaller than) the $P_\eta$ enclosures, and (2) continuous transformations between the $P_\eta$. When $j'$ and $k'$ are inside the same basin, this corresponds to the first case. It can be shown that there is one or more hidden enclosures of dimension $\mu_\sim < rank(P_\eta)$. 

When $j'$ and $k'$ are inside different basins, there will be some continuous transformation of commuting projectors, but not necessarily between $P_{\eta_j}$ and $P_{\eta_k}$. If the equivalence classes have one element each in $N_{\eta_j}$ and $N_{\eta_k}$, then the two basins are unitarily equivalent. Otherwise, one or both basins contain smaller hidden enclosures, and there is one (or possibly more) continuous transformation between hidden enclosures.
\end{remark}

\begin{remark}
From a practical standpoint, let us consider how to determine the existence of the required function $f_C$, given an arbitrary $GPM$-Lindblad system. There are two steps: (1) constructing $f_C$ using some subset of the phases of the Lindblad operators and (2) determining whether the remaining phases satisfy the $f_C$.

Even though the function has a domain of size $|N_C |^2 = O(\nu_\sim^4)$, the transitivity means the function only has $|N_C| - 1 = O(\nu_\sim^2)$ degrees of freedom. For some ordering of nodes $(j_1,k_1)$, $\dots$, $(j_{|N_C|},k_{|N_C|})$, we only need to store the values $f((j_1,k_1),(j_2,k_2))$, $f((j_1,k_1),(j_3,k_3))$, $\dots$, $f((j_{1},k_{1}), (j_{|N_C|},k_{|N_C|}))$. It is always possible to draw a non-self-intersecting path that intersects all nodes in some order, so we can compute the above function values in $|N_C|-1$ steps in a recursive manner:
\begin{align*}
f_C((j_{1}, k_{1}), (j_l,k_l)) &= f_C((j_{1}, k_{1}), (j_{l-1},k_{l-1})) \times \\
& f_C((j_{l-1}, k_{l-1}), (j_l,k_l)) \\
&\hspace{-0.8in} = f_C((j_{1}, k_{1}), (j_{l-1},k_{l-1})) \times \\ & \hspace{-0.5in}\exp\left(i\left(\arg(D_{\alpha_{l-1}, j_{l-1}})-\arg(D_{\alpha_{l-1}, k_{l-1}})\right)\right).
\end{align*}

The second step means taking the remaining edges, and ensuring they obey the phase function. For an arbitrary edge $(j_{l_1},k_{l_1})\rightarrow_{\alpha}(j_{l_2},k_{l_2})$, this means computing the equation:
\begin{align*}
&\exp\left(i\left(\arg(D_{\alpha,j_{l_1}})-\arg(D_{\alpha,k_{l_1}})\right)\right) \\
&\hspace{0.5in}= f_C((j_{l_1},k_{l_1}),(j_{l_2},k_{l_2})) \\
&\hspace{0.5in}= f_C((j_{1},k_{1}),(j_{l_1},k_{l_1}))^* f_C((j_{1},k_{1}),(j_{l_2},k_{l_2})).
\end{align*}

There are most $N|N_C|$ edges that need to be checked. 
Overall, there are at most $|N_C| - 1 + N|N_C| = O(N|N_C|) = O(N\nu_\sim^2)$ steps to this process.

Instead of checking the phases of a given set of Lindblad operators, a related problem is to write down how the subset of dependent phases (\emph{i.e.} those involved in the second step) depend on the free phases (\emph{i.e.} the phases involved in the first step). This is a linear algebra-type problem with a twist. The first equations that one writes down are in the form $e^{iAx} = e^{ib}$, where the $x$ is the vector of dependent phases. Of course when one takes the logarithm of both sides, one has a linear equation modulo $2\pi$. This means that there are multiple discrete solutions, as a term of $\frac{2\pi k}{n'}$ will appear. $n'$ will be some integer that depends on the connectivity of $G_\sim$. 
\end{remark}

The theorem covers the case of hidden enclosures, each of which has a unique full-rank asymptotic state. There may also be oscillations between these states, whether they are between hidden enclosures or the generic enclosures represented by the $P_\eta$'s. Say we have projectors $P_a$ and $P_b$, where either can be a hidden enclosure, or correspond to a basin $P_\eta$. These projectors have corresponding equivalence relations $\sim_a$ and $\sim_b$ on domain $N_a$ and $N_b$ respectively: if it corresponds to a hidden enclosure, it is the equivalence relation required by the theorem; otherwise, it is the trivial equivalence relation (\emph{i.e.} each equivalence class is a singleton).
\begin{corollary}
There are stable oscillations between dissipation blocks $a$ and $b$ if there is an equivalence relation $\asymp$ on $N_a\cup N_b$ such that
\begin{enumerate}
\item $\asymp$ restricted to $N_a$ or $N_b$ is identical to $\sim_a$ and $\sim_b$ respectively.
\item $\asymp$ obeys the dissipation symmetry and resonance properties of the theorem.
\item There is a constant $\Delta$ such that, for any $j\asymp k$, $j\in N_a$, $k\in N_b$, we have $H_j-H_k=\Delta$.
\end{enumerate}
\end{corollary}

\begin{proof}
The key here is that oscillations in general Lindblad systems only exist when the Hamiltonian-less system has a continuous transformation between the dissipation blocks. In particular, the operator $P_{ab}:=\frac{1}{2}\left(P_a+P_b+U+U^\dagger\right)$ is a projector, where $U$ is the unitary operator from block $a$ to block $b$:
\begin{align*}
P_{ab}^2 &=
\frac{1}{4}\left(P_a^2+P_b^2+P_aP_b+P_bP_a + U^2+U^{\dagger 2}+UU^\dagger \right. \\
& \left. + U^\dagger U + P_a(U+U^\dagger)+ P_b(U+U^\dagger) \right. \\
&\left. + (U+U^\dagger)P_a+ (U+U^\dagger)P_b\right) \\
&= \frac{1}{4}\left(P_a + P_b + P_b + P_a + U^\dagger + U + U + U^\dagger\right) \\
&= P_{ab},
\end{align*}
and it commutes with the Lindblad operators (but not necessarily the Hamiltonian). The theorem guarantees the existence of the equivalence relation $\asymp$ on $N_{ab}:=N_a\cup N_b$, such that the dissipation symmetry and resonance properties are satisfied (but not necessarily the Hamiltonian symmetry). Since the projector is proportional to $P_a$ and $P_b$ when restricted to $N_a$ and $N_b$, respectively, the relation $\asymp$ when similarly restricted must match $\sim_a$ and $\sim_b$.

All that is left is to prove the last condition. Oscillations on Lindblad systems must satisfy the commutation of $U+U^\dagger$ with $H-h_aP_a-h_bP_b$ for some constants $h_a$, $h_b$, which means that the following must vanish:
\begin{align*}
[U+U^\dagger, H-h_aP_a-h_bP_b] &= [U+U^\dagger, H] + \\
& \hspace{0.2in} (h_a-h_b)(U-U^\dagger) \\
& \hspace{-1in} = [P_{ab}, H]  + (h_a-h_b)(U-U^\dagger),
\end{align*}
where we have used the fact that $P_a$ and $P_b$ must commute with $H$. Because of the equivalence relation, $P_{ab}$ only has non-zero elements $P_{ab,jk}$ when $j\asymp k$:
\begin{align*}
[P_{ab},H] &= \sum_{j,k \in N_{ab}, j\asymp k} P_{ab,jk}(H_j-H_k)|j\rangle\langle k|.
\end{align*}
For arbitrary $j,k$ with $j\asymp k$, $j\in N_a$,  $k\in N_b$, there is a sequence $j\asymp l_1 \asymp l_2 \dots \asymp l_f \asymp k$ such that $P_{ab,jl_1}, P_{ab,l_1l_2}, \dots P_{ab,l_fk} \ne 0$. For each pair $(l',l")$ in this sequence, we have $H_{l'}-H_{l"} = (h_b-h_a)(\delta_{l', a}-\delta_{l", a})$, where the delta is one if and only if the respective element is in $N_a$. Adding these up, we get $H_j-H_k = (H_j-H_{l_1}) + (H_{l_1}-H_{l_2}) + \dots + (H_{l_f}-H_{k}) = h_b-h_a =: \Delta$, and we have shown the final condition.
\end{proof}

\section{Examples}

\subsection{Single Basin Example}

Let us consider two $n=4$ Lindblad operators:
\begin{align*}
L_1 &= |1\rangle\langle 2| + 2|2\rangle\langle 3| + 3|3\rangle\langle 4| + 4|4\rangle\langle 1| \\
L_2 &= 5|1\rangle\langle 3| + 10|3\rangle\langle 2| + 2|2\rangle\langle 4| + 6|4\rangle\langle 1| \\
\end{align*}
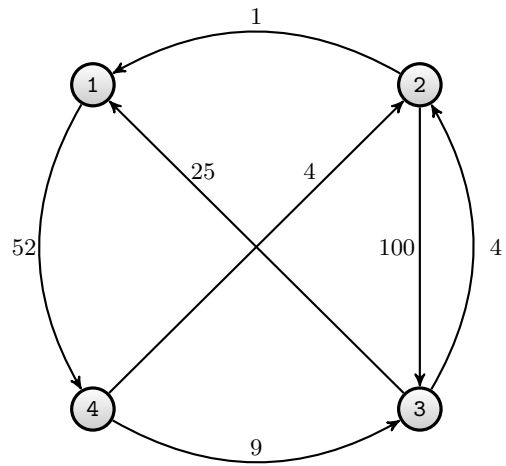
\begin{figure}
	\centering
	\begin{tikzpicture}[->,node distance=3mm,
	vertex/.style={
		circle,minimum size=4mm,rounded corners=2mm,
		very thick,draw=black,
		top color=white,bottom color=black!20,
		font=\ttfamily}, >=stealth]
	
	\node (1) [vertex] {1};
	\node (2) [vertex, right of=1, xshift=40mm]{2};
	\node (3) [vertex, below of=2, yshift=-40mm]{3};
	\node (4) [vertex, left of=3, xshift=-40mm]{4};

	\path  (3) edge[->, >=stealth',thick] node[xshift=-7mm,yshift=10mm]{25} (1)
	(4) edge[->, >=stealth',thick] node[xshift=7mm, yshift=10mm]{4} (2)
	(2) edge[->, >=stealth',thick] node[xshift=-3mm]{100} (3);
	\path  (1) edge[->, >=stealth',thick, bend right] node[xshift=-2mm]{52} (4)
	(4) edge[->, >=stealth',thick, bend right] node[yshift=2mm]{9} (3)
	(3) edge[->, >=stealth',thick, bend right] node[xshift=3mm]{4} (2)
	(2) edge[->, >=stealth',thick, bend right] node[yshift=2mm]{1} (1);

	\end{tikzpicture}
	\caption{$G_\Omega$ for $L_1$ and $L_2$ under the diagonal projectors.}
	\label{ex1GOm}
\end{figure}
$G_\Omega$ is shown in Fig. \ref{ex1GOm}. This graph is strongly connected, and there is clearly no equivalence relation $\sim$ that would satisfy the Lindblad operators. Therefore, there is a single asymptotic state that every initial state will approach (we have not specified the Hamiltonian, but assuming it is diagonal, this will not change the asymptotic state). 

In general, for $n=4$, there are $4^3=64$ trees, with $16$ rooted at each vertex. Our graph has only seven out of twelve possible edges, which gives fifteen total trees (six rooted at one, three each for the other three vertices). This gives the following $w$-products:
\begin{itemize}
\item Rooted at 1: $w_{12}w_{23}w_{34}=36$, $w_{12}w_{23}w_{24}=16$, $w_{12}w_{24}w_{13}=100$, $w_{12}w_{13}w_{34}=225$, $w_{13}w_{32}w_{24}=10,000$, $w_{13}w_{32}w_{34}=22,500$. 
\item Rooted at 2: $w_{23}w_{34}w_{41}=1872$, $w_{24}w_{41}w_{13}=5,200$, $w_{24}w_{41}w_{23}=20,800$.
\item Rooted at 3: $w_{34}w_{41}w_{12}=468$, $w_{34}w_{41}w_{32}=46,800$, $w_{32}w_{24}w_{41}=20,800$.
\item Rooted at 4: $w_{41}w_{12}w_{23}=208$, $w_{41}w_{12}w_{13}=1,300$, $w_{41}w_{13}w_{32}=130,000$.
\end{itemize}
The asymptotic state is then:
\begin{align*}
\frac{1}{260,325}\left(32,877|1\rangle\langle 1| + 27,872|2\rangle\langle 2| + 68,068|3\rangle\langle 3| \right. \\
\left. + 131,508|4\rangle\langle 4| \right)
\end{align*}

\subsection{Forest constraints}

Now consider two $n=8$ Lindblad operators:
\begin{align*}
\sigma_3&:= (12)(34)(5678) \\
D_3&:=diag(2,3,4,5,6,7,8,9) \\
\sigma_4&:= (15)(26)(37)(48) \\
D_4&:=diag(0,0,0,0,10,11,12,13)
\end{align*}
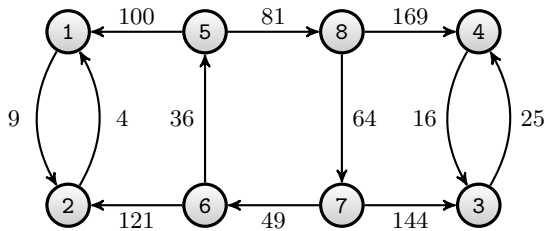
\begin{figure}
	\centering
	\begin{tikzpicture}[->,node distance=3mm,
	vertex/.style={
		circle,minimum size=4mm,rounded corners=2mm,
		very thick,draw=black,
		top color=white,bottom color=black!20,
		font=\ttfamily}, >=stealth]
	
	\node (5) [vertex] {5};
	\node (6) [vertex, below of=5, yshift=-20mm]{6};
	\node (7) [vertex, right of=6, xshift=15mm]{7};
	\node (8) [vertex, above of=7, yshift=20mm]{8};
	\node (1) [vertex, left of=5, xshift=-15mm]{1};
	\node (2) [vertex, left of=6, xshift=-15mm]{2};
	\node (3) [vertex, right of=7, xshift=15mm]{3};
	\node (4) [vertex, right of=8, xshift=15mm]{4};

	\path  (8) edge[->, >=stealth',thick] node[xshift=3mm,yshift=0mm]{64} (7)
	(7) edge[->, >=stealth',thick] node[xshift=0mm, yshift=-2mm]{49} (6)
	(6) edge[->, >=stealth',thick] node[xshift=-3mm, yshift=0mm]{36} (5)
	(5) edge[->, >=stealth',thick] node[yshift=2mm]{81} (8);
	\path  (1) edge[->, >=stealth',thick, bend right] node[xshift=-3mm]{9} (2)
	(2) edge[->, >=stealth',thick, bend right] node[xshift=3mm]{4} (1);
	\path  (4) edge[->, >=stealth',thick, bend right] node[xshift=-3mm]{16} (3)
	(3) edge[->, >=stealth',thick, bend right] node[xshift=3mm]{25} (4);
	\path  (5) edge[->, >=stealth',thick] node[yshift=2mm]{100} (1)
	(6) edge[->, >=stealth',thick] node[yshift=-2mm]{121} (2);
	\path  (7) edge[->, >=stealth',thick] node[yshift=-2mm]{144} (3)
	(8) edge[->, >=stealth',thick] node[yshift=2mm]{169} (4);
	\end{tikzpicture}
	\caption{$G_\Omega$ for $L_3$ and $L_4$ under the diagonal projectors.}
	\label{ex2GOm}
\end{figure}
$G_\Omega$ is shown in Fig. \ref{ex2GOm}. Clearly, $\{1,2\}$ and $\{3,4\}$ are basins, as they are strongly connected with no outgoing edges. $\{5,6,7,8\}$ is also strongly connected but has outgoing edges, and therefore this subgraph corresponds to the decay space. Asymptotically, the system will leave this space.

In each basin, there is a full-rank asymptotic state (again we ignore the Hamiltonian while assuming it is diagonal). There is no symmetry in either basin, so there are no hidden enclosures. Hence the asymptotic state in either basin is unique. These states are:
\begin{align*}
\rho_{12}:= \frac{1}{13}(9|1\rangle\langle 1|+4|2\rangle\langle 2|) \\
\rho_{34}:= \frac{1}{41}(16|3\rangle\langle 3|+25|4\rangle\langle 4|).
\end{align*} 
The coefficients correspond to trees with only one edge. 

While the trees are straightforward, the forests that determine the constraints are somewhat less so. We need to find forests with roots in $\{1,2,3,4\}$. We discard the edges within basins as this speeds up the calculation. There are fifteen such forests: there are $2^4=16$ choices for outgoing edges $\{5,6,7,8\}$, and one of those choices corresponds to the cycle. One example forest is shown in figure 	\ref{ex2GOmfor}. The corresponding $w$-product is $w_{15}w_{56}w_{67}w_{48} = 29,811,600$. The remaining products are $w_{15}w_{26}w_{37}w_{48}$, $w_{15}w_{26}w_{37}w_{78}$, $w_{15}w_{26}w_{67}w_{48}$, $w_{15}w_{56}w_{37}w_{48}$, $w_{85}w_{26}w_{37}w_{48}$, $w_{15}w_{26}w_{67}w_{78}$, $w_{15}w_{56}w_{37}w_{78}$, $w_{85}w_{26}w_{37}w_{78}$, $w_{85}w_{26}w_{67}w_{48}$, $w_{85}w_{56}w_{37}w_{48}$, $w_{15}w_{56}w_{67}w_{78}$, $w_{85}w_{26}w_{67}w_{78}$, $w_{85}w_{56}w_{37}w_{78}$, and $w_{85}w_{56}w_{67}w_{48}$.
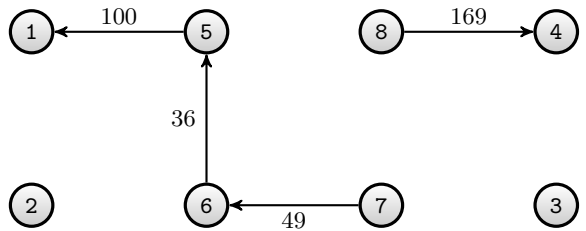
\begin{figure}
	\centering
	\begin{tikzpicture}[->,node distance=3mm,
	vertex/.style={
		circle,minimum size=4mm,rounded corners=2mm,
		very thick,draw=black,
		top color=white,bottom color=black!20,
		font=\ttfamily}, >=stealth]
	
	\node (5) [vertex] {5};
	\node (6) [vertex, below of=5, yshift=-20mm]{6};
	\node (7) [vertex, right of=6, xshift=20mm]{7};
	\node (8) [vertex, above of=7, yshift=20mm]{8};
	\node (1) [vertex, left of=5, xshift=-20mm]{1};
	\node (2) [vertex, left of=6, xshift=-20mm]{2};
	\node (3) [vertex, right of=7, xshift=20mm]{3};
	\node (4) [vertex, right of=8, xshift=20mm]{4};

	\path  (7) edge[->, >=stealth',thick] node[xshift=0mm, yshift=-2mm]{49} (6)
	(6) edge[->, >=stealth',thick] node[xshift=-3mm, yshift=0mm]{36} (5);
	\path  (5) edge[->, >=stealth',thick] node[yshift=2mm]{100} (1);
	\path (8) edge[->, >=stealth',thick] node[yshift=2mm]{169} (4);
	\end{tikzpicture}
	\caption{One forest in the set $T_B(G_\Omega')$, where $\Omega$ is formed from $L_3$ and $L_4$ with diagonal projectors.}
	\label{ex2GOmfor}
\end{figure}
To calculate the constraint for $\kappa_1'$, where $N_1:=\{1,2\}$, we add the forests where $l\in\{5,6,7,8\}$ leads to that basin. For example, there are $3$ forests where $l=8$ leads to node $1$ or $2$:
\begin{align*}
\sum_{\tau\in T_B(G'_\Omega, 1, 8)} W(\tau) &= w_{15}w_{26}w_{67}w_{78} + w_{26}w_{67}w_{78}w_{85}\\ 
&\hspace{.8in} + w_{15}w_{56}w_{67}w_{78} \\
&= 8,718,080.
\end{align*}
The $\kappa'_\eta$ vectors are:
\begin{align*}
\kappa'_1 &= (140073700,140073700,140073700,140073700,\\ & 74395890, 125013820, 31739260, 8718080)^T \\
\kappa'_2 &= (140073700,140073700,140073700,140073700, \\
& 65677810, 26211240, 108334440, 111443300)^T 
\end{align*}

\subsection{Hidden Enclosures}
Now let us consider a system with hidden enclosures. Let us define two $n=9$ Lindblad operators:
\begin{align*}
\sigma_5&:= (147)(258)(369) \\
D_{5} &:= diag(4e^{i\theta_{5,1}},5e^{i\theta_{5,2}},6e^{i\theta_{5,3}},4e^{i\theta_{5,4}},5e^{i\theta_{5,5}},6e^{i\theta_{5,6}},\\
&\hspace{1in} 4e^{i\theta_{5,7}},5e^{i\theta_{5,8}},6e^{i\theta_{5,9}})\\
\sigma_6&:= (123)(456)(789) \\
D_{6} &:= diag(e^{i\theta_{6,1}},2e^{i\theta_{6,2}},3e^{i\theta_{6,3}},e^{i\theta_{6,4}},2e^{i\theta_{6,5}},3e^{i\theta_{6,6}}, \\
&\hspace{1in}e^{i\theta_{6,7}},2e^{i\theta_{6,8}},3e^{i\theta_{6,9}})
\end{align*}

\begin{figure}
	\centering
	\begin{tikzpicture}[->,node distance=3mm,
	vertex/.style={
		circle,minimum size=4mm,rounded corners=2mm,
		very thick,draw=black,
		top color=white,bottom color=black!20,
		font=\ttfamily}, >=stealth]
	
	\node (1) [vertex] {1};
	\node (2) [vertex, right of=1, xshift=20mm]{2};
	\node (3) [vertex, below right of=1, xshift=10mm, yshift=-10mm]{3};

	\node (4) [vertex, right of=2, xshift=20mm]{4};
	\node (5) [vertex, right of=4, xshift=20mm]{5};
	\node (6) [vertex, below right of=4, xshift=10mm, yshift=-10mm]{6};

	\node (7) [vertex, below of=2, yshift=-35mm]{7};
	\node (8) [vertex, below of=4, yshift=-35mm]{8};
	\node (9) [vertex, above right of=7, xshift=10mm, yshift=10mm]{9};

	\path  (3) edge[->, >=stealth',thick] node[xshift=-2mm]{9} (2)
	(2) edge[->, >=stealth',thick] node[yshift=2mm]{4} (1)
	(1) edge[->, >=stealth',thick] node[xshift=3mm]{1} (3);

	\path  (6) edge[->, >=stealth',thick] node[xshift=-1mm,yshift=1mm]{9} (5)
	(5) edge[->, >=stealth',thick] node[yshift=2mm]{4}(4)
	(4) edge[->, >=stealth',thick] node[xshift=1mm,yshift=1mm]{1}(6);

	\path  (9) edge[->, >=stealth',thick] node[xshift=-1mm,yshift=-1mm]{9} (8)
	(8) edge[->, >=stealth',thick] node[yshift=-2mm]{4} (7)
	(7) edge[->, >=stealth',thick] node[xshift=1mm,yshift=-1mm]{1}(9);

	\path  (7) edge[->, >=stealth',thick, out=80, in=210] 	node[xshift=1mm,yshift=-4mm]{16} (4)
	(4) edge[->, >=stealth',thick, bend right]	node[yshift=2mm]{16} (1)
	(1) edge[->, >=stealth',thick]	node[xshift=-2mm, yshift=-2mm]{16} (7);

	\path  (8) edge[->, >=stealth',thick] node[xshift=1mm,yshift=-2mm]{25} (5)
	(5) edge[->, >=stealth',thick, bend right]node[yshift=2mm]{25} (2)
	(2) edge[->, >=stealth',thick, out=-30, in=100] node[xshift=-1mm,yshift=-4mm]{25} (8);

	\path  (9) edge[->, >=stealth',thick, bend right=8mm] node[xshift=1mm,yshift=-2mm]{36} (6)
	(6) edge[->, >=stealth',thick] node[yshift=-2mm]{36} (3)
	(3) edge[->, >=stealth',thick, bend right=8mm]node[xshift=-4mm,yshift=0mm]{36} (9);
	\end{tikzpicture}
	\caption{$G_\Omega$ for $L_5$ and $L_6$ under the diagonal projectors.}
	\label{ex3GOm}
\end{figure}
$G_\Omega$ is shown in Fig. \ref{ex3GOm}. Clearly, the entire graph is strongly connected, so there no decay space. This graph has one basin, but does that mean there is only one enclosure for the entire system? 

The Lindblad operators have a symmetry, and there is an equivalence relation $1\sim 4\sim 7 \not\sim 2\sim 5\sim 8\not\sim 3\sim 6 \sim 9\not\sim 1$. First of all, $|D_{\alpha,1+\delta}|=|D_{\alpha,4+\delta}|=|D_{\alpha,7+\delta}|$, for $\alpha=1,2$ and $\delta=0,1,2$. Both permutations also preserve the equivalence structure, but not in the same manner. $\sigma_6$ cycles between the three classes, while $\sigma_5/\sim$ is the identity permutation on the quotient space. 

So there will be a hidden enclosure if we can satisfy Hamiltonian degeneracy and dissipation resonance. Let us suppose we have the following Hamiltonian, which has the appropriate degeneracy:
\begin{align*}
H = |1\rangle \langle 1|  + |4\rangle \langle 4|+ |7\rangle \langle 7| - |3\rangle \langle 3|-|6\rangle \langle 6|-|9\rangle \langle 9|.
\end{align*}

Now let us consider resonance. One component of $G_\sim$ is shown in Fig. \ref{ex3Gsim} (weights not shown, to reduce clutter). There are two components total, but the other is the image of the first under exchange of indices. If one is resonant, so is the other, hence we only consider the first component. 

Note that both operators are required for connectivity. Each one by itself would give three smaller components of three nodes each, and the requirement for resonance would be looser. In fact, because $L_5$ maps equivalence classes to themselves, the resonance condition is automatically satisfied. It is only because $L_6$ connects each of its components that there are conditions imposed on its phases.

There are nine nodes, so we need eight degrees of freedom to construct the function $f_C$:
\begin{align*}
f_C((1,4), (7,1)) &= \exp(i(\theta_{5,1}-\theta_{5,4}))\\
f_C((7,1), (4,7)) &= \exp(i(\theta_{5,7}-\theta_{5,1}))\\
f_C((4,7), (6,9)) &= \exp(i(\theta_{6,4}-\theta_{6,7}))\\
f_C((6,9), (3,6)) &= \exp(i(\theta_{5,6}-\theta_{5,9}))\\
f_C((3,6), (9,3)) &= \exp(i(\theta_{5,3}-\theta_{5,6}))\\
f_C((9,3), (8,2)) &= \exp(i(\theta_{6,9}-\theta_{6,3}))\\
f_C((8,2), (5,8)) &= \exp(i(\theta_{5,8}-\theta_{5,2}))\\
f_C((5,8), (2,5)) &= \exp(i(\theta_{5,5}-\theta_{5,8})).
\end{align*}
All other function values are fixed via transitivity. To construct $f_C$, we have used six solid edges, and two dashed edges. There are three other solid edges, but they automatically satisfy resonance. We have seven remaining edges to check. For example:
\begin{align*}
&\exp(i(\theta_{6,1}-\theta_{6,4})) = f_C((1,4), (3,6)) \\
& \hspace{.5in} = f_C((1,4), (7,1))f_C((7,1), (4,7)) \\
& \hspace{.8in} \times f_C((4,7), (6,9))f_C((6,9), (3,6)) \\
& \hspace{.5in} = \exp(i(\theta_{5,7}+\theta_{6,4}+\theta_{5,6}-\theta_{5,4}-\theta_{6,7}-\theta_{5,9})).
\end{align*}
We can write down seven linear equations that the phases must satisfy (all equations $\mod 2\pi$):
\begin{align*}
\theta_{6,1}-\theta_{6,4} &= \theta_{5,7}+\theta_{6,4}+\theta_{5,6}-\theta_{5,4}-\theta_{6,7}-\theta_{5,9} \\
\theta_{6,3}-\theta_{6,6} &= \theta_{5,3} + \theta_{6,9}  + \theta_{5,5} -\theta_{5,6}-\theta_{6,3}-\theta_{5,2} \\
\theta_{6,2}-\theta_{6,5} &= \theta_{6,4}+\theta_{6,6}-\theta_{6,1}-\theta_{6,3} \\
\theta_{6,6}-\theta_{6,9} &= \theta_{5,6} + \theta_{6,3} + \theta_{5,8} - \theta_{5,9} - \theta_{6,6} - \theta_{5,5}  \\
\theta_{6,5}-\theta_{6,8} &= \theta_{5,5} + \theta_{6,2} + \theta_{5,7} - \theta_{5,8} - \theta_{6,5} - \theta_{5,4} \\
\theta_{6,8}-\theta_{6,2} &= \theta_{5,8} + \theta_{6,5} + \theta_{5,1} - \theta_{5,2} - \theta_{6,8} - \theta_{5,7}  \\
\theta_{6,7}-\theta_{6,1} &= \theta_{5,7} + \theta_{6,4} + \theta_{5,3} - \theta_{5,1} - \theta_{6,7} - \theta_{5,9}.
\end{align*}
There are seven equations, but only six degrees of freedom, since we can add a constant to the phases $\theta_{6,j}$ and get another solution. There is a further redundancy due to the fact these are linear equations $\mod 2\pi$. If we have a solution for $\theta_{5,j}$ and $\theta_{6,j}$, we can find another solution by adding $\frac{2\pi k}{3}$ to $\theta_{6,4}$, $\theta_{6,5}$, and $\theta_{6,6}$, and subtracting $\frac{2\pi k}{3}$ from $\theta_{6,7}$, $\theta_{6,8}$, and  $\theta_{6,9}$, for $k=1$ or $2$. 

\begin{figure}
	\centering
	\begin{tikzpicture}[node distance=3mm,
	vertex/.style={
		circle,minimum size=2mm,rounded corners=1mm,
		very thick,draw=black,
		top color=white,bottom color=black!20,
		font=\ttfamily}, >=stealth]
	
	\node (1) [vertex]{(1,4)};
	\node (2) [vertex, right of=1, xshift=70mm]{(3,6)};
	\node (3) [vertex, below left of=2, xshift=-35mm, yshift=-70mm]{(2,5)};
	\node (4) [vertex, above right of=3, xshift=10mm, yshift=15mm]{(8,2)};
	\node (5) [vertex, above left of=3, xshift=-10mm, yshift=15mm]{(5,8)};
	\node (6) [vertex, below right of=1, xshift=17.5mm, yshift=-2.5mm]{(4,7)};
	\node (7) [vertex, below left of=6, xshift=-10mm, yshift=-10mm]{(7,1)};
	\node (8) [vertex, below left of=2, xshift=-17.5mm, yshift=-2.5mm]{(9,3)};
	\node (9) [vertex, below right of=8, xshift=10mm, yshift=-10mm]{(6,9)};
	
	\path  (1) edge[->, >=stealth',thick, dashed, bend left] (2)
	(2) edge[->, >=stealth',thick, dashed, bend left] (3)
	(3) edge[->, >=stealth',thick, dashed, bend left] (1);

	\path  (4) edge[->, >=stealth',thick, dashed] (7)
		(7) edge[->, >=stealth',thick, dashed] (8)
		(8) edge[->, >=stealth',thick, dashed] (4);
	
	\path  (6) edge[->, >=stealth',thick, dashed] (9)
		(9) edge[->, >=stealth',thick, dashed] (5)
		(5) edge[->, >=stealth',thick, dashed] (6);
	
	\path  (3) edge[->, >=stealth',thick] (4)
		(4) edge[->, >=stealth',thick] (5)
		(5) edge[->, >=stealth',thick] (3);

	\path  (1) edge[->, >=stealth',thick] (7)
		(7) edge[->, >=stealth',thick] (6)
		(6) edge[->, >=stealth',thick] (1);

	\path  (2) edge[->, >=stealth',thick] (8)
		(8) edge[->, >=stealth',thick] (9)
		(9) edge[->, >=stealth',thick] (2);
	
	\end{tikzpicture}
	\caption{One component of $G_\sim$ for the operators $L_5$ and $L_6$. The solid edges pertain to $L_5$ and the dashed to $L_6$. Weights not shown.}
	\label{ex3Gsim}
\end{figure}

Define the following function $g:N_C' \rightarrow U(1)$, where $N_C':= N_C \cup \{(j,j): j\in \mathbb{Z}_9\}$:
\begin{align*}
g(j,j) & := 1\\
g(j,\sigma_5(j))&:= \exp\left(i\left(-\theta_{5,j}+\frac{1}{3}\sum_{k=0}^2 \theta_{5,\sigma_5^k(j)} \right)\right) \\
g(\sigma_5(j),j)&:= g(j,\sigma_5(j))^*  \\
j&:= 1,\dots,9
\end{align*}
Note that this function is transitive in the same sense as $f_C$: for any $j\sim k \sim l$, $g(j,k)g(k,l)=g(j,l)$.

Now define a basis of nine vectors, for $j=1,2,3$, $k=0,1,2$:
\begin{align*}
|j,k\rangle &:= \frac{1}{\sqrt{3}}\sum_{l=0}^2 g(j+3l,j)e^{2\pi ikl/3}|j+3l\rangle.
\end{align*}
This is an orthonormal basis:
\begin{align*}
\langle j, k|j',k'\rangle &= \frac{1}{3}\sum_{l,l'=0}^2 g(j+3l,j)^*g(j'+3l',j') \times \\
& \hspace{0.5in}e^{2\pi i(-kl+k'l')/3}\langle j+3l|j'+3l'\rangle \\
&= \frac{1}{3}\delta_{jj'} \sum_{l=0}^2 g(j+3l,j)^*g(j+3l,j)\times \\
& \hspace{0.5in}e^{2\pi i(-k+k')l/3} \\
&= \frac{1}{3}\delta_{jj'} \sum_{l=0}^2 e^{2\pi i(-k+k')l/3} \\&= \delta_{jj'}\delta_{kk'}.
\end{align*}
Furthermore in this basis, we have the following:
\begin{align*}
H &= \sum_{k=0}^2\left(|1,k\rangle\langle 1,k| - |3,k\rangle\langle 3,k|\right) \\
L_5 &= \sum_{k=0}^2 e^{2\pi ik/3} \left(4e^{i(\theta_{5,1}+\theta_{5,4}+\theta_{5,7})} |1,k\rangle\langle 1,k|  \right.\\
&\hspace{0.8in}+5e^{i(\theta_{5,2}+\theta_{5,5}+\theta_{5,8})} |2,k\rangle\langle 2,k| \\
&\hspace{0.8in}+6e^{i(\theta_{5,3}+\theta_{5,6}+\theta_{5,9})} |3,k\rangle\langle 3,k| \Big) \\
L_6 &= 2e^{i\theta_{6,2}} \left( |1,0\rangle\langle 2,1|+|1,1\rangle\langle 2,2|+|1,2\rangle\langle 2,0|\right) \\
&+ 3e^{i\theta_{6,3}}\left( |2,0\rangle\langle 3,1|+|2,1\rangle\langle 3,2|+|2,2\rangle\langle 3,0|\right) \\
&+ e^{i\theta_{6,1}}\left( |3,0\rangle\langle 1,1|+|3,1\rangle\langle 1,2|+|3,2\rangle\langle 1,0|\right).
\end{align*}
In this basis, we have a GPM-system: the Hamiltonian is diagonal, and the Lindblad operators are GPM's. However, now there are three basins instead of one (see Fig. \ref{ex3GOmBeta}). 

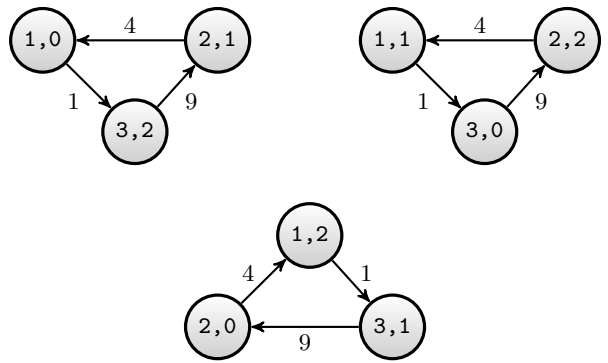
\begin{figure}
	\centering
	\begin{tikzpicture}[->,node distance=3mm,
	vertex/.style={
		circle,minimum size=4mm,rounded corners=2mm,
		very thick,draw=black,
		top color=white,bottom color=black!20,
		font=\ttfamily}, >=stealth]
	
	\node (1) [vertex] {1,0};
	\node (2) [vertex, right of=1, xshift=20mm]{2,1};
	\node (3) [vertex, below right of=1, xshift=10mm, yshift=-10mm]{3,2};

	\node (4) [vertex, right of=2, xshift=20mm]{1,1};
	\node (5) [vertex, right of=4, xshift=20mm]{2,2};
	\node (6) [vertex, below right of=4, xshift=10mm, yshift=-10mm]{3,0};

	\node (7) [vertex, below of=2, yshift=-35mm]{2,0};
	\node (8) [vertex, below of=4, yshift=-35mm]{3,1};
	\node (9) [vertex, above right of=7, xshift=10mm, yshift=10mm]{1,2};

	\path  (3) edge[->, >=stealth',thick] node[xshift=2mm,yshift=-2mm]{9} (2)
	(2) edge[->, >=stealth',thick] node[yshift=2mm]{4} (1)
	(1) edge[->, >=stealth',thick] node[xshift=-2mm,yshift=-2mm]{1} (3);

	\path  (6) edge[->, >=stealth',thick] node[xshift=2mm,yshift=-2mm]{9} (5)
	(5) edge[->, >=stealth',thick] node[yshift=2mm]{4} (4)
	(4) edge[->, >=stealth',thick] node[xshift=-2mm,yshift=-2mm]{1} (6);

	\path  (9) edge[->, >=stealth',thick] node[xshift=2mm,yshift=1mm]{1} (8)
	(8) edge[->, >=stealth',thick] node[yshift=-2mm]{9} (7)
	(7) edge[->, >=stealth',thick] node[xshift=-2mm,yshift=1mm]{4} (9);

	\end{tikzpicture}
	\caption{$G_\Omega$ for $L_5$ and $L_6$ after transforming to the $|j,k\rangle$ basis.}
	\label{ex3GOmBeta}
\end{figure}

In the transformed system, one cannot find hidden enclosures. There is no symmetry inside the basins, since the edge weights are all different. There is a symmetry between the basins, but the existence of $L_5$ ruins the resonance of the new $G_\sim$. Fig. \ref{ex3GsimBeta} shows one component of the new $G_\sim$. There are now three nodes; moreover, each node has a self-edge with a weight that is not one. It follows that this component cannot be resonant, since a transitive function must be one when the two input nodes are identical. The other five components of $G_\sim$ also have self-edges where the weight are $e^{\pm 2\pi i/3}$. 

\begin{figure}
	\centering
	\begin{tikzpicture}[node distance=3mm,
	vertex/.style={
		circle,minimum size=2mm,rounded corners=1mm,
		very thick,draw=black,
		top color=white,bottom color=black!20,
		font=\ttfamily}, >=stealth]

	\node (1) [vertex]{(1,0),(2,1)};
	
	\node (2) [vertex, below right of=1, xshift=30mm, yshift=-40mm]{(3,2),(1,0)};
	\node (3) [vertex, left of=2, xshift=-60mm]{(2,1),(3,2)};
	
	\path  (1) edge[->, >=stealth',thick, dashed, bend left] (2)
	(2) edge[->, >=stealth',thick, dashed, bend left] (3)
	(3) edge[->, >=stealth',thick, dashed, bend left] (1);

	\path  (1) edge[->, >=stealth',thick, loop above] node[yshift=1mm]{$e^{-2\pi i/3}$} (1);
	\path  (2) edge[->, >=stealth',thick, loop below]node[yshift=-1mm]{$e^{-2\pi i/3}$} (2);
	\path  (3) edge[->, >=stealth',thick, loop below]node[yshift=-1mm]{$e^{-2\pi i/3}$} (3);

	\end{tikzpicture}
	\caption{One component of $G_\sim$ for the operators $L_5$ and $L_6$, after transforming to the $|j,k\rangle$ basis. The solid edges pertain to $L_5$ and the dashed to $L_6$. Weights not shown.}
	\label{ex3GsimBeta}
\end{figure}
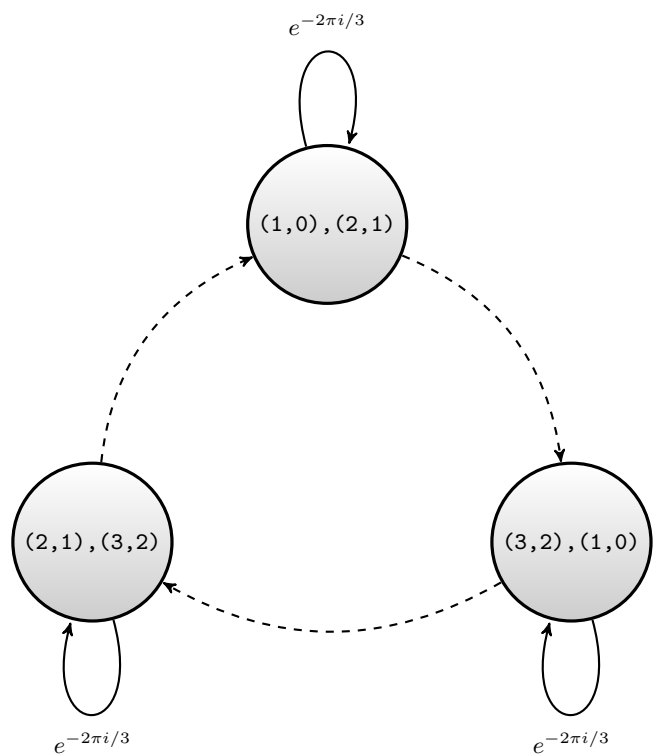
So we can conclude that each basin in the new $G_\Omega$ is a minimal enclosure, and that any asymptotic state must be a linear superposition of the unique dissipative states for each basin. The dissipative states for each basin will have the same co-efficients. There are only three trees for each basin, one rooted at each node. These trees will have weights $4$, $9$, and  $36$.

So, the system must approach one of the following states:
\begin{align*}
\rho &= \frac{c_1}{49}\left(36|1,0\rangle\langle 1,0|+9|2,1\rangle\langle 2,1|+4|3,2\rangle\langle 3,2| \right) \\
&+\frac{c_2}{49}\left(36|1,1\rangle\langle 1,1|+9|2,2\rangle\langle 2,2|+4|3,0\rangle\langle 3,0| \right)\\
 &+ \frac{c_3}{49}\left(36|1,2\rangle\langle 1,2|+9|2,0\rangle\langle 2,0|+4|3,1\rangle\langle 3,1| \right) \\
 \sum_{j=1}^3 c_j&=1.
\end{align*}


\section{Conclusions}

We have derived combinatoric formulas for stationary orbits of the quantum density operator under the influence of the Lindblad super-operator. These formulas relate Lindblad operators to trees and forests on a directed graph $G_\Omega$. Each node on the graph represents a (one-dimensional) eigenspace of the density operator, and the edge weights represent (modulus-squared) off-diagonal matrix elements of the Lindblad operators. The stationary orbits have non-zero eigenvalues on the basins of the graph (the strongly-connected components that have no outgoing edges). Each basin has a unique stationary orbit whose coefficients are derived from the trees in the basin rooted at the relevant node.

The piece of the density operator outside of the basins will vanish if the relevant eigenprojectors are held stationary. In this case, one can write down linear constraints which can be used to calculate the asymptotic end-state given an initial state. These constraints have coefficients that depend on the forests on graph $G_\Omega$, where the trees in the forests have root nodes in the basins.

The interplay between projectors and orbits is thorny for the general case: finding stationary eigenprojectors can be difficult. However, we have examined a class of Lindblad operators that take the form of generalized permutation matrices. The diagonal projectors in these systems are stationary, and therefore our tree and forest formulas can be used to describe the asymptotic end-states of the system. The diagonal projectors are typically the only stationary projectors, but if there is a particular symmetry on the Lindblad and Hamiltonian operators, there may be other end-states: either hidden enclosures, or oscillations between enclosures. This symmetry is described by an equivalence relation on a subset of the nodes of $G_\Omega$, where the Hamiltonian amplitudes are symmetric, the Lindblad amplitudes are symmetric, and there is a resonance between the phases of the Lindblad operators.


\begin{acknowledgments}
The work of P. Rooney was supported by Discovery Grant \#311880 of the Natural Science and Engineering Research Council of Canada, and the DFG Grant HE 1858/13-1 from the German Research Foundation (DFG).  The work of A. M. Bloch was supported by grants NSF DMS 1613819 and AFOSR 9550-18-0028. The work of C. Rangan was supported by Discovery Grant \#311880 of the Natural Science and Engineering Research Council of Canada.
\end{acknowledgments}

\bibliographystyle{unsrt}
\bibliography{QCbiblio}

\newcommand{\noop}[1]{}
\begin{thebibliography}{10}

\bibitem{NielsenChuangBook}
M.~A. Nielsen and I.~L. Chuang.
\newblock {\em Quantum Computation and Quantum Information}.
\newblock Cambridge University Press, 2000.

\bibitem{DAlessandroBook}
D.~D'Alessandro.
\newblock {\em Introduction to Quantum Control and Dynamics}.
\newblock Chapman \& Hall/CRC, 2008.

\bibitem{BreuerPetruccioneBook}
H.-P. Breuer and F.~Petruccione.
\newblock {\em The Theory of Open Quantum Systems}.
\newblock Oxford University Press, 2007.

\bibitem{Lindblad76}
G.~Lindblad.
\newblock On the generators of quantum dynamical semigroups.
\newblock {\em Comm. Math. Phys.}, 48:119, 1976.

\bibitem{GoriniKossakowskiSudarshan76}
V.~Gorini, A.~Kossakowski, and E.C.G. Sudarshan.
\newblock Completely positive dynamical semigroups of ${N}$-level systems.
\newblock {\em J. Math. Phys.}, 17(5):821, 1976.

\bibitem{LloydViola01}
S.~Lloyd and L.~Viola.
\newblock Engineering quantum dynamics.
\newblock {\em Phys. Rev. A}, 65:010101, 2001.

\bibitem{Baconetal01}
D.~Bacon et~al.
\newblock Universal simulation of {M}arkovian quantum dynamics.
\newblock {\em Phys. Rev. A}, 64:062302, 2001.

\bibitem{Barreiroetal2011}
J.T. Barreiro et~al.
\newblock An open-system quantum simulator with trapped ions.
\newblock {\em Nature}, 470:486, 2011.

\bibitem{Spohn1976}
H.~Spohn.
\newblock Approach to equilibrium for completely positive dynamical semigroups
  of {N}-level systems.
\newblock {\em Rep. Math. Phys.}, 10:189--94, 1976.

\bibitem{Frigerio1978}
A.~Frigerio.
\newblock Stationary states of quantum dynamical semigroups.
\newblock {\em Commun. Math. Phys.}, 63:269--76, 1978.

\bibitem{SchirmerWang2010}
S.~Schirmer and X.~Wang.
\newblock Stabilizing open quantum systems by markovian reservoir engineering.
\newblock {\em Physical Review A}, 81:062306, 2010.

\bibitem{BaumgartnerNarnhofer2012}
B.~Baumgartner and H.~Narnhofer.
\newblock {The structures of state space concerning {Q}uantum {D}ynamical
  {S}emigroups}.
\newblock {\em Rev. Math. Phys.}, 24({1250001}), 2012.

\bibitem{us_nis2_b}
P.~Rooney, A.M. Bloch, and C.~Rangan.
\newblock Flag-based control of quantum purity for $n=2$ systems.
\newblock {\em Physical Review A}, 93(6), 2016.

\bibitem{us_FBC}
P.~Rooney, A.M. Bloch, and C.~Rangan.
\newblock Steering the eigenvalues of the density operator in
  {H}amiltonian-controlled quantum {L}indblad systems.
\newblock {\em IEEE Transactions on Automatic Control}, 63(6), 2018.

\bibitem{HarrisHirstBook}
J.~Harris, J.L. Hirst, and M.~Mossinghoff.
\newblock {\em Combinatorics and Graph Theory, 2nd ed.}
\newblock Springer, 2008.

\bibitem{ChaikenKleitman1978}
S.~Chaiken and D.J. Kleitman.
\newblock Matrix tree theorems.
\newblock {\em J. Comb. Theory Series A}, 24:377--381, 1978.

\bibitem{Shietal2016}
I.R.~Petersen G.D.~Shi, D.Y.~Dong and K.H. Johansson.
\newblock Reaching a quantum consensus: Master equations that generate
  symmetrization and synchronization.
\newblock {\em IEEE Trans. Aut. Cont.}, 61(2):374--387, 2016.

\bibitem{KatoBook}
T.~Kato.
\newblock {\em Perturbation Theory for Linear Operators}.
\newblock Springer-Verlag, 1980.

\bibitem{BaumgartnerNarnhofer2008}
B.~Baumgartner and H.~Narnhofer.
\newblock Analysis of quantum semigroups with gks–lindblad generators: {II}.
  general.
\newblock {\em J. Phys. A: Math. Theor.}, 41:395303, 1978.

\bibitem{Cayley}
A.~Cayley.
\newblock A theorem on trees.
\newblock {\em Quart. J, Math}, 23:376, 1889.

\bibitem{Chaiken1982}
S.~Chaiken.
\newblock A combinatorial proof of the all minors matrix tree theorem.
\newblock {\em SIAM J. Alg. Disc. Meth.}, 3(3):319--329, 1982.

\end{thebibliography}

\end{document}